\documentclass[10pt,shortpaper,twoside,web] {ieeecolor}

\usepackage{generic}
\usepackage{cite}
\usepackage{amsmath,amssymb,amsfonts} 
\usepackage{algorithmic}
\usepackage{textcomp}
\usepackage{psfrag} 
\usepackage{cuted}
\usepackage{float}
\usepackage{graphicx} 
\usepackage{amsmath}  

\usepackage{graphicx}
\usepackage{wrapfig}
\usepackage{caption}

\usepackage[fancythm,fancybb,morse,ieee]{jphmacros2e} 
\usepackage{array}
\usepackage{bbm}
\usepackage{mathrsfs}
\usepackage{epstopdf}
\DeclareMathAlphabet{\mathpzc}{OT1}{pzc}{m}{it}
\usepackage{graphicx}
\usepackage{amsmath}
\usepackage[fancythm,fancybb,morse,ieee]{jphmacros2e}
\usepackage{caption}
\usepackage{subcaption}
\usepackage{multirow}
\usepackage{csquotes}
\usepackage{url}
\newtheorem{thm}{Theorem}
\newtheorem{lem}{Lemma}
\newtheorem{prop}{Proposition}

\newtheorem{defn}{Definition}

\def\BibTeX{{\rm B\kern-.05em{\sc i\kern-.025em b}\kern-.08em
		T\kern-.1667em\lower.7ex\hbox{E}\kern-.125emX}}
\begin{document}
	\title{Data-based Low-conservative Nonlinear Safe Control Learning}
	\author{Amir Modares,  Bahare Kiumarsi, and Hamidreza Modares
		\thanks{{A. Modares is with the University of Cyprus, Nicosia, Cyprus (e-mail: amir.modares.81@gmail.com, modares.amir1360@yahoo.com). B. Kiumarsi and H. Modares are with Michigan State University, East Lansing, MI, USA (e-mails:kiumarsi@msu.edu, modaresh@msu.edu).}}
        }
	
\maketitle
\thispagestyle{empty} 
\begin{abstract} 
This paper develops a data-driven safe control framework for nonlinear discrete-time systems with parametric uncertainty and additive disturbances. The proposed approach constructs a data-consistent closed-loop representation that enables controller synthesis and safety certification directly from data.
Unlike existing methods that treat unmodeled nonlinearities as global worst-case uncertainties using Lipschitz bounds, the proposed approach embeds nonlinear terms directly into the invariance conditions via a geometry-aware {convexification based on quadratic curvature compensation}. This enables facet- and direction-specific convexification, avoiding both nonlinearity cancellation and the excessive conservatism induced by uniform global bounds. {We further propose a node-dependent interpolated controller together with a cell-wise, direction-resolved curvature slack, which enforces the convexity conditions locally on the cells of a partition of the safe set while every facet remains certified at every point, thereby enlarging the class of certifiable invariant sets.} For systems subject to additive disturbances, disturbance effects are embedded directly into the verification conditions through optimized, geometry-dependent bounds, rather than via uniform margin inflation, yielding less conservative robust safety guarantees. As a result, the proposed methods can certify substantially larger safe sets, naturally accommodate joint state and input constraints, and provide data-driven safety guarantees. The simulation results show a significant improvement in both nonlinearity tolerance and the size of the certified safe set.
\end{abstract}

\IEEEpeerreviewmaketitle

\section{Introduction}

\IEEEPARstart{C}ertification of the safety of autonomous control systems is essential for their deployment in real-world applications. 
Accordingly, such systems are designed to achieve performance objectives while guaranteeing safety, a combination that often leads to challenging and potentially intractable control problems. 
Several approaches have been proposed to address this challenge: 
1) safety filters \cite{SAfeRL4,SAfeRL5,SAfeRL10,SAfeRL16} based on control barrier functions (CBFs) \cite{SB4,CBFN1,SB6,SB7,SB8}, 
2) control merging approaches \cite{merge1} that combine safe and goal-reaching controllers, and 
3) model predictive control (MPC), which solves constrained optimal control problems in a receding-horizon fashion \cite{MPC1}. 

While CBF-based control design can be efficiently applied to nonlinear continuous-time (CT) systems, its extension to discrete-time (DT) systems generally leads to non-convex optimization problems \cite{safeDTconvex}. 
Control merging approaches require learning multiple controllers and coordinating their interaction, which becomes particularly challenging for nonlinear systems under uncertainty and disturbances. 
MPC provides a systematic framework for constraint handling but suffers from increased computational complexity for nonlinear systems and typically relies on conservative terminal invariant sets when linear terminal controllers are used.

{Despite significant progress in safety filters, maximal admissible sets \cite{MAS,MASpolytopic}, motion planning \cite{SepHypLanding,SepHypExclusion,SepHypSafeTubes,SepHypRigidTube} and MPC-based methods \cite{PolyControl,TubeConfig,TubeRigid,TubeHomothetic,TubeElastic}, existing safe control approaches for nonlinear discrete-time systems typically maintain convexity by either canceling nonlinear terms or bounding their effects using global worst-case envelopes.} While these strategies lead to tractable formulations, they often ignore useful structure in the nonlinear dynamics and consequently yield overly conservative controllers and small certifiable safe sets. This motivates the need for a different design paradigm that preserves convexity while directly accounting for nonlinear dynamics in the safety conditions, rather than eliminating or overbounding them. 

 Another line of research studies controlled invariant sets for nonlinear discrete-time systems using reachability-based methods, including sum-of-squares (SOS) relaxations \cite{CIS1,CIS2} and successive convexification via local linearization \cite{CIS3}. These approaches often treat uncertainties conservatively, rely on computationally intensive iterative procedures, and introduce control only through existential quantification without synthesizing an explicit feedback law. While invariance conditions for polyhedral sets in certain nonlinear systems are derived in \cite{DC}, no controller synthesis procedure or nonconservative design method is provided.

To deal with uncertainties, data-driven safe control methods are commonly used and are classified as indirect and direct. 
Indirect approaches first identify a system model from data and then design safety filters or controllers, including data-driven safety filters \cite{Data6,Data7,Data8,Data66,Data62,Data63,Data9} and data-driven MPC methods \cite{MPCd1,MPCd2,MPCd3,MPCd4,MPCd5}. 
However, these approaches prioritize predictive accuracy rather than control objectives, which can result in conservative or suboptimal performance. Direct approaches bypass explicit system identification and learn controllers directly from data \cite{Data1,Data2,Data3,Data3a,Data4,DataNon1,DataNon2,Data5}. 
While mostly limited to linear systems, in \cite{Data4,Data3a}, nonlinear effects are suppressed or treated as disturbances during controller learning, yielding conservative designs under safety constraints. Similarly, CBF- and contractivity-based methods \cite{DataNon1,DataNon2} restrict allowable nonlinearities and require heavy optimization, limiting their scalability.

This paper introduces a direct data-driven safe control framework for nonlinear
discrete-time systems with parametric uncertainty and additive disturbances.
{In safety-critical applications, the polytopic safe set is
typically prescribed by the application, but such a template is only a
{candidate}: it is generally not control invariant, and a controller
certifying it must be synthesized. Existing methods certify only small subsets of
a prescribed template, because they retain convexity by canceling the nonlinear
terms or by absorbing them into global worst-case envelopes. The main
contribution of this paper is a set of certificates that instead retain the
nonlinear terms inside the invariance conditions and thereby certify
substantially more of the same template---directly from data, with unknown
system matrices.} {This is achieved through} a synthesis pipeline
that {convexifies the closed-loop facet maps by quadratic curvature compensation},
imposing contractivity without cancelling nonlinearities or absorbing them into
conservative disturbance bounds; invariance is instead guaranteed by explicitly
accounting for how the nonlinear terms bend the facet maps. When a single global
{curvature compensation} remains conservative under strong nonlinearities, we propose a
hybrid approach that exactly convexifies part of the facet map and bounds only
the remainder by coarse estimates, and we show it strictly outperforms existing
nonlinearity minimization approaches. To further improve feasibility and enlarge
the class of certifiable invariant sets, we {replace the scalar
curvature slack by direction-resolved matrix slacks assigned to the cells of a
partition of the safe set and develop node-dependent interpolated controller
parameterizations,}
significantly reducing conservatism{, while allowing joint state
and input constraints to be imposed explicitly}. The framework extends to
additive disturbances by retaining disturbance-dependent terms directly in the
safety conditions rather than collapsing them into worst-case bounds.
{On the same prescribed template, the proposed certificates
tolerate nonlinearities up to {$188\times$} larger and certify invariant sets up to
{$4.1\times$} larger than the Lipschitz-based baseline, and the gains are retained
in higher dimensions}.\\

 \vspace{-6pt}

\noindent \textbf{Notations and Definitions.} Throughout, $\mathbb{R}^n$ and $\mathbb{R}^{n\times m}$ denote the sets
of real vectors and real matrices of dimensions $n$ and $n\times m$,
respectively, and $\mathbb{R}_{\ge 0}$ denotes the set of nonnegative real
numbers. For a matrix $P$, $P \ge 0$ ($P\le0$) indicates element-wise
nonnegativity (nonpositivity), while $P \succeq 0$ ($P \preceq 0$) denotes that $P$ is positive (negative)
semi-definite. For vectors $x,y\in\mathbb{R}^n$, $x\le y$ denotes
element-wise inequality. The norms $\|\cdot\|$, $\|\cdot\|_1$, and $\|\cdot\|_2$ denote the $\infty$, $1$,
and Euclidean norms, respectively. For a vector $L \in \mathbb{R}^n$, $L_i$ is its $i$-th entry. For a matrix $X\in\mathbb{R}^{n\times m}$, $x_{ij}$ denotes its $(i,j)$-th
entry, and $X_{i,:}$ and $X_{:,i}$ its $i$-th row and column, respectively. The symbols $I$ and $\mathbf 1$ denote, respectively, the identity matrix and the all-ones vector of appropriate dimensions. For a set $\cal S$,
 $\mathrm{int}\cal S$, 
$\partial\cal S$, and $\mathrm{vert}(\cal S)$  denote its interior, boundary, and the set of its vertices,
respectively. $\mathrm {diag}(\cdot)$ denotes a diagonal matrix, and $\odot$ denotes the elementwise product.

\begin{defn}[{\textbf{Constraint, safe, and contractive sets \cite{SetB}}}]
\label{def:sets}
{%
Let $\mathcal{X}=\{x:Fx\le g\}$ denote the state constraint set, imposed a priori
by the physical limitations of the system or by its operating environment. A set $\mathcal{S}\subseteq\mathcal{X}$ is called
\begin{enumerate}
  \item a safe set (i.e., a robust invariant set {(RIS)}) for a given closed-loop map $x(t+1)=f(x(t),w(t))$
        if $x(0)\in\mathcal{S}\Rightarrow x(t)\in\mathcal{S}$ for all $t\ge0$ and
        all $w(t)\in\mathcal{W}$.
  \item {$\lambda$-contractive}  if, given $\lambda\in (0 \quad 1]$, $x(t)\in\mathcal{S}$ implies that
$x(t+1)\in \lambda\mathcal{S}$ for all $t\ge 0$ and for all $w(t)\in\cal{W}$. {If $\mathcal S$ is convex and $0\in\mathcal S$, then $\lambda\mathcal S\subseteq\mathcal S$, so a $\lambda$-contractive set is an RIS.}
\end{enumerate}
A raw constraint set $\mathcal{X}$ is, in general, {not} control invariant;
control invariance is a property that must be certified through control design.}
\end{defn}

\vspace{6pt}

\begin{defn} \textbf{Convex Functions} \cite[Chapter~3]{convex}:
A function $f:\mathbb{R}^n \rightarrow \mathbb{R}$
is convex on a convex set $\Omega$ if for all $x, y \in \Omega$, and all $\lambda \in [0 \quad 1]$, we have $f\big(\lambda x+(1-\lambda) y \big) \le \lambda f(x)+ (1-\lambda) f(y).
$ The second-order sufficient condition for convexity is given by 
 $\frac{\partial f^2(x)}{\partial x^2}   \succeq 0, \,\, \forall x \in \Omega.$

\end{defn} \vspace{6pt}



\begin{defn} \textbf{ Polytopic Sets}
\cite[Chapter~3]{SetB}: A polyhedral set $\cal{P}(F,g)$ is defined as $\cal{S}(F,g)=\{x\in\mathbb{R}^n \mid Fx\le g\}$, where $F\in\mathbb{R}^{s\times n}$ and $g\in\mathbb{R}^s$.
Every bounded polyhedral set is referred to as a
{polytope}.
A {face} of a polyhedral set is the intersection of the set with a
supporting hyperplane. Vertices are faces of dimension~$0$, and a
facet is a maximal proper face.
Let $\mathcal I_s=\{1,\dots,s\}$ denote the index set of facet constraints.
Let $\mathrm{vert}(\cal S(F,g))=\{x_1,\dots,x_\ell\}$, and define their index set  as
$\mathcal I_\ell = \{1,\dots,\ell\}$. {When $\cal S(F,g)$ is covered by a finite family of polytopes (cells) $\{\mathcal C_c\}$, $\mathcal I_c=\{1,\dots,n_c\}$ denotes the index set of the cells, $\mathcal J$ the index set of the nodes, i.e.\ of the vertices of all cells, and $\mathcal J_c\subseteq\mathcal J$ the index set of the vertices of the cell $\mathcal C_c$.}
For any $x\in\cal S(F,g)$ define the active set
$\mathcal I_f(x) = \{i\in\mathcal I_s: F_{i,:}x=g_i\},$
and its minimal face
$\mathcal F(x) = \{z\in\cal S(F,g): F_{i,:}z=g_i,\ \forall i\in\mathcal I_f(x)\}.$ For any face $\mathcal F$ of $\cal S(F,g)$, let $\mathcal I_f(\mathcal F)$ denote its active facets
(i.e., $i\in\mathcal I_f(\mathcal F)$ if and only if $F_{i,:}z=g_i$ for all $z\in\mathcal F$).
\end{defn} \vspace{3pt}

\section{Problem Formulation}

\subsection{Problem Setup}
Consider the discrete-time nonlinear system (DT-NS) as 
\begin{equation}\label{system} 
x(t+1) = {A}Z(x(t)) + Bu(t) + w(t),
\end{equation}
where $x(t) \in \mathbb{R}^n$ is the system's state, $u(t) \in \mathbb{R}^m$ is the control input, and $w(t) \in \cal{W} \subset \mathbb{R}^n$ is the additive disturbance. Moreover,
\begin{align} \label{Snon}
    Z(x(t)) = \begin{bmatrix} x(t)^\top & S(x(t))^\top \end{bmatrix}^\top \in \mathbb{R}^{N+n},
\end{align}
where $S(x(t)) \in \mathbb{R}^{N}$ is a {twice continuously
differentiable} dictionary of nonlinear functions. Using \eqref{Snon}, the system parameters can be decomposed into
\begin{align} \label{A1A2}
    A=[A_1 \quad A_2] \in \mathbb{R}^{n \times {(n+N)}},
\end{align}
where $A_1 \in \mathbb{R}^{n \times {n}}$ and $A_2 \in \mathbb{R}^{n \times {N}}$ are, respectively, the linear and nonlinear parts of the system model. \vspace{3pt}


\begin{assumption}
    The system matrices $A$ and $B$ are unknown, while $Z(x)$ is known. Moreover, $x=0$ is an equilibrium.
\end{assumption} \vspace{3pt}

\begin{assumption}
   The disturbance set satisfies the bound
\vspace{-6pt}
\begin{align}
\label{hw}
\cal{W}=\{\,w\in\mathbb{R}^n \mid \|w\|\le h_w\,\}.
\end{align}
\end{assumption}
{
\begin{assumption}\label{safeset}
The system's {safe-set template} is a prescribed polytope
\begin{equation} \label{eq:safeset}
  \cal{S}(F,g)=\{x\in\mathbb{R}^{n}\;:\;Fx\le g\},
  \, F\in\mathbb{R}^{s\times n},\; g\in\mathbb{R}^{s},
\end{equation}
The set
$\cal{S}(F,g)$ is a candidate set that we require the closed-loop
trajectory to remain inside{, and $g\ge0$, i.e., $0\in\cal S(F,g)$, so that $\lambda\cal S(F,g)\subseteq\cal S(F,g)$ for $\lambda\in(0,1]$}.
\end{assumption}
}
\vspace{3pt}


{
\begin{remark}\label{rem:template}
Essentially all existing safety-related work prescribes a fixed template,
for at least three distinct reasons.

{(i) Inherent physical constraints:} The template is frequently dictated
by the application rather than by the dynamics: box constraints on a robot's
joint positions, velocities, or heading angles give a hyper-rectangle
$\cal{S}=\{x:|x_i|\le r_i\}$, and operational envelopes (voltage, current, temperature) are
naturally polytopic. In all such cases $F$ is inherently fixed, while
the admissible size (through $g$, or a facet scaling) and the certifying controller remain the design
freedoms.

{(ii) Computational tractability:} Computing {maximal} (robust) control
invariant sets exactly is, in general, intractable: exact vertex/facet
enumeration and Minkowski operations on general polytopes are
NP-hard~\cite{PolyControl,PolyNPhard}. The standard approach is to fix the set {complexity}---a facet
matrix $F$, a vertex count, or a polynomial degree---and to optimize the
controller and the remaining set parameters within that class. This underlies
maximal-output-admissible-set constructions with a fixed
feedback~\cite{MAS,MASpolytopic}, the configuration-constrained polytopes of the
polyhedral-control literature~\cite{PolyControl,TubeConfig}, and the fixed
cross-section shapes of 
tubes in MPC~\cite{TubeRigid,TubeHomothetic,TubeElastic}. The same choice is made in
reachability analysis, where the representation complexity is fixed
to retain scalability~\cite{CIS3,convexinterpol,DC,CIS1,CIS2}. 

{(iii) Convexity in motion planning:} The obstacle-free region is inherently
nonconvex, and the standard remedy is to prescribe polytopic safe sets obtained
by locally convexifying the avoidance constraints via separating hyperplanes, so
that the resulting synthesis remains convex; see the next motivating example and
Figure~\ref{fig:safetube}, and
\cite{SepHypLanding,SepHypExclusion,SepHypSafeTubes,SepHypRigidTube}.

 \vspace{3pt}

\end{remark}}

{%
\textbf{Motivating Example.} {The nonconvex admissible region in Fig.~\ref{fig:safetube} can be covered by a sequence of prescribed polytopic safe sets $\cal{S}_k=\mathcal X\cap\bigcap_j \mathcal H_{k,j}\subseteq\mathcal X$, whose facets are determined by the state constraints and the separating hyperplanes used to exclude obstacles. For each $\cal S_k$, a $\lambda$-contractive controller is designed to preserve safety while tracking the nominal path toward the overlap with the next cell \cite{SepHypLanding}. Expressing $\cal S_k$ in coordinates centered at an equilibrium $x^{\rm e}_k\in\mathrm{int}(\cal S_k\cap\cal S_{k+1})$ gives $0\in\mathrm{int}\,\cal S_k$ (Assumption~\ref{safeset}), and the decrease of $V$ in Corollary~\ref{cor:polyhedral_lyap_relaxed} drives the state toward $x^{\rm e}_k$, hence into $\cal S_k\cap\cal S_{k+1}$ in finite time. Once the state enters $\cal S_k\cap\cal S_{k+1}$, control is handed over to the controller associated with $\cal S_{k+1}$. Thus, invariance guarantees safety within each prescribed cell, while the tracking and handover mechanism enables progression along the route. This template-based construction is common in safe-corridor planning and convex MPC, where the environment determines the polyhedral geometry and the certifying controller remains the primary design freedom.}} \vspace{3pt}

\begin{figure}[t]
  \centering
\includegraphics[width=\columnwidth,height=2.4in]{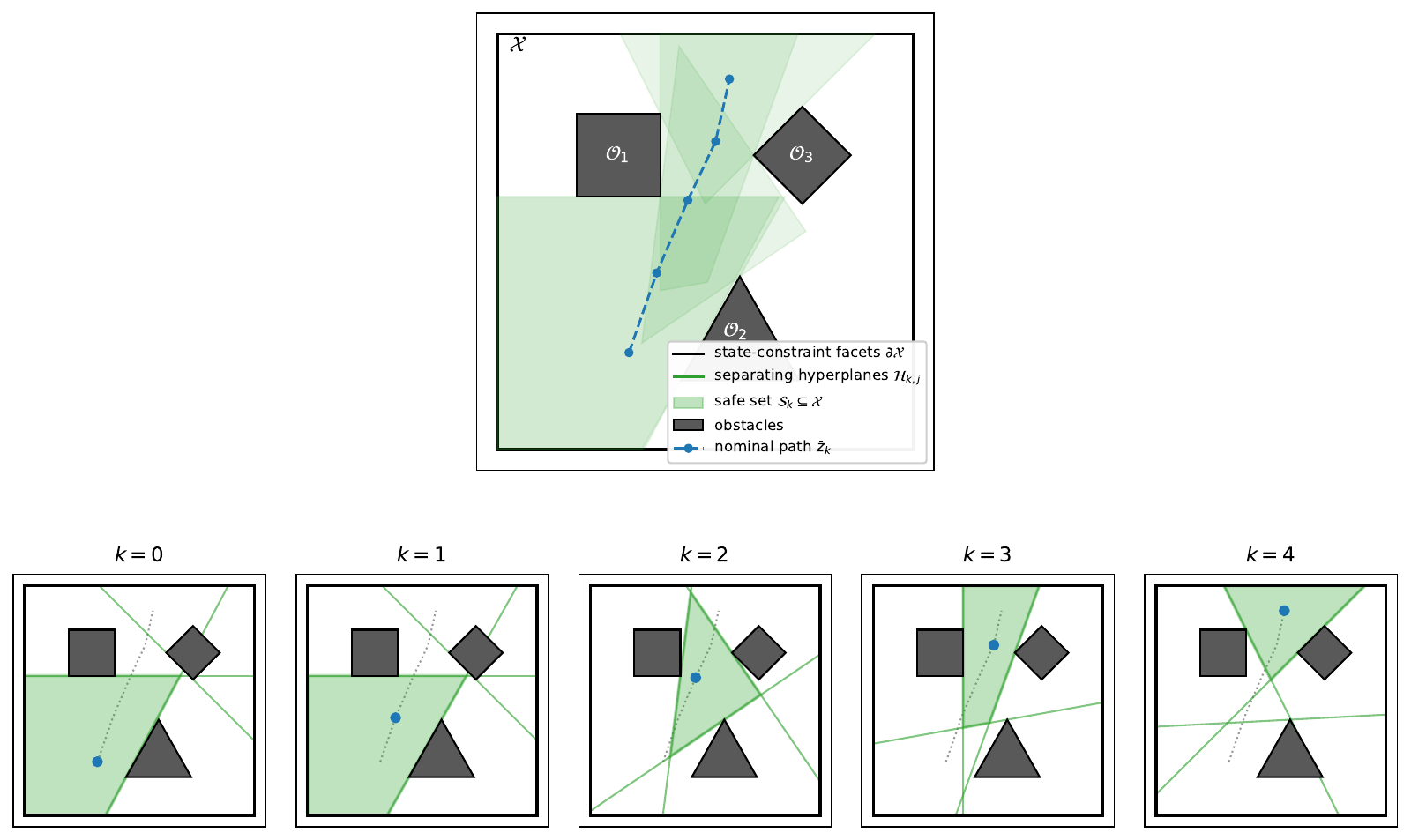}
  \setlength{\abovecaptionskip}{-3pt}
  \caption{{Safe navigation motivates fixed polyhedral templates.
  The state-constraint box $\mathcal{X}$ minus the obstacles (dark gray) is
  nonconvex, so it is broken into prescribed safe sets
  $\cal{S}_k=\mathcal{X}\cap\bigcap_j\mathcal{H}_{k,j}(\bar z_k)\subseteq\mathcal{X}$
  (green), each a closed, compact polytope whose facets are the state-constraint
  facets $\partial\mathcal{X}$ (black) and the active separating hyperplanes
  $\mathcal{H}_{k,j}$ (green lines) that exclude the obstacles. The dashed blue
  path is the nominal trajectory $\bar z_k$. The template reshapes as the state
  advances.}}
  \label{fig:safetube}
\end{figure}

Since $A$ and $B$ are unknown, the controller is learned from data by applying an input sequence given by
\eqref{system}
\begin{align} \label{data-u}
U_0 = \begin{bmatrix} u(0) & u(1) & \cdots & u(T-1) \end{bmatrix} \in \mathbb{R}^{m \times T}.
\end{align}
Collect $T+1$ samples of the state vectors as 
$X = \begin{bmatrix} x(0) & x(1) & \cdots & x(T) \end{bmatrix} \in \mathbb{R}^{n \times (T+1)}$.
These collected samples are then organized as follows 
\begin{align} \label{data-x}
X_0 &= \begin{bmatrix} x(0) & x(1) & \cdots & x(T-1) \end{bmatrix} \in \mathbb{R}^{n \times T}, \\ \label{data-xx}
X_1 &= \begin{bmatrix} x(1) & x(2) & \cdots & x(T) \end{bmatrix} \in \mathbb{R}^{n \times T}.
\end{align}

Besides, the sequence of unmeasurable disturbances is 
\begin{equation}\label{data}
W_0 = \begin{bmatrix} w(0) & w(1) & \ldots & w(T-1) \end{bmatrix} \in \mathbb{R}^{n \times T},
\end{equation}


\begin{remark}
In \cite{Data4}, a nonlinear control policy of the form
\begin{align} \label{cont}
u(t)=KZ(x(t))=K_1x(t)+K_2 S(x(t))
\end{align}
is used to design stabilizing controllers.
With the dynamics \eqref{system}-\eqref{A1A2},
this choice yields the closed-loop system
$x(t+1)=(A_1+BK_1)x(t)+(A_2+BK_2)S(x(t))+w(t)$.
To preserve tractability, the approach in \cite{Data4} treats the nonlinear term $(A_2+BK_2)S(x_k)$ as a disturbance and seeks to minimize its effect, typically by enforcing small norms of $A_2+BK_2$.
\end{remark}
\vspace{6pt}

Since $S(x)$ is known, define 
\begin{align} \label{As}
    A_s=\frac{\partial S(x(t))}{\partial x}\big|_{x=0} \in \mathbb{R}^{N \times n},
\end{align}
We now propose a controller in the form of
\begin{align} \label{contNew}
    u(t)=K_1x(t)+K_2 Q(x(t)),  
\end{align}
where \vspace{-6pt}
\begin{align} \label{Q}
Q(x(t))=S(x(t))-A_s x(t),
\end{align}
is the nonlinearity remainder. Using \eqref{system} and \eqref{contNew}, the closed-loop dynamics  becomes  
\begin{align} \label{clNew}
    x(t+1)=(\bar A_1+BK_1) x(t)+ (A_2+BK_2) Q(x(t))+w(t),
\end{align} \vspace{-9pt}
where \vspace{-6pt}
\begin{align} \label{barA}
  \bar A_1=A_1+A_2A_s.
\end{align}
The control law \eqref{contNew}, and hence the closed-loop system, consists of a linear term and a nonlinear remainder from linearization. Compared to \eqref{cont}, this interleaved linear--nonlinear structure is jointly optimized to enhance safety. {We work with $Q$ rather than $S$ because $\partial Q/\partial x|_{x=0}=0$: all first-order effects of the dictionary are collected in $\bar A_1$, so $K_1$ acts on the true linearization $(\bar A_1,B)$ (Assumption~4), while $Q$ retains only higher-order terms whose Lipschitz and curvature bounds on $\cal S(F,g)$, which drive the certificates below, are not inflated by the linear part of $S$. Since $A_2S(0)=0$ at the equilibrium $x=0$ (with $u=0$), we take $S(0)=0$ without loss of generality, so $Q(0)=0$.}
Define \vspace{-6pt}
\begin{align} \label{data-z}
    V_0 = \begin{bmatrix} X_0 \\ Q(X_0) \end{bmatrix} \in \mathbb{R}^{(n+N) \times T}.
\end{align}
where
\begin{align} \label{data-Q}
Q(X_0) &= \begin{bmatrix} Q(x(0)) & \cdots & Q(x(T-1)) \end{bmatrix} \in \mathbb{R}^{N \times T}.
\end{align}
The following assumptions are required. \vspace{6pt}

\begin{assumption}
   The pair $(\bar A_1,B)$ is stabilizable. 
\end{assumption} \vspace{3pt}

\begin{assumption}\label{assumption_5}
The data matrix $V_0$ has full row rank, and the number of samples satisfies $T \geq n+N+1$.
\end{assumption} \vspace{6pt}


\begin{assumption} \label{lipass}
     $Q(x)$ in \eqref{Q} is Lipschitz on $\cal S(F,g)$, i.e., $\|Q(x)-Q(x_0)\|_2\le L_Q\|x-x_0\|_2$ $\forall x,x_0\in\cal S(F,g)$ and some {$L_Q>0$}. Given that the nonlinear function $Q(x) \in \mathbb{R}^{N}$ is known, {componentwise} Lipschitz bounds on $\cal S(F,g)$ can be obtained as $L_{Q,j}\ge 0$ such that, for all $x,y\in\cal S(F,g)$,
\begin{alignat}{2}
\label{eq:comp_lip_assm}
\big|Q_j(x)-Q_j(y)\big|
\le L_{Q,j}\,\|x-y\|_2,\qquad j=1,\ldots,N .
\end{alignat}
Then, $Q(x)$ is Lipschitz on $\cal S(F,g)$, and for any
$r\in\mathbb R^{N}$,
${|}r^\top(Q(x)-Q(y)){|}\le \|L_Q\odot r\|_1\|x-y\|_2$, where $L_Q=[L_{Q,1},\ldots,L_{Q,N}]^\top$. 
\end{assumption} \vspace{6pt}


\begin{remark}
Assumption~1 requires a known nonlinear structure, which is standard \cite{Data66,Data3a,Data4,DataNon1,DataNon2}. Assumptions~2 and 3 specify bounded disturbances and a polytopic safe set. Assumption~4 ensures the controller parameterization is flexible enough to admit at least one stabilizing feedback policy. Assumption~5 requires rich data for closed-loop learning, a condition weaker than identifying $A$ and $B$ directly (Remark~3). Assumption~6 imposes a Lipschitz bound on the nonlinearities, which is necessary for the nonlinearity minimization approach.
\end{remark} \vspace{3pt}

\noindent \textbf{Problem 1: Data-based Safe control Design:} 
Consider the system \eqref{system} with collected input--state data given by
\eqref{data-u}--\eqref{data-xx}, and suppose that Assumptions~1--6 hold. Fix $\lambda\in (0 \quad 1]$. The
objective is to learn a nonlinear safe control policy of the form \eqref{contNew}
such that the safe set {\eqref{eq:safeset}} is an RIS. \vspace{-6pt}

{%
\begin{remark}
\label{rem:onestep}
 A myopic approach to safety is to minimize a convex stage
cost $\ell(x,u)$ while enforcing one-step state constraints:
\begin{align*}
  u^{\star}(x)\in\arg\min_{u\in\mathcal{U}}\ &\ell(x,u)\\
  \text{s.t.}\quad & F\bigl(AZ(x(t))+Bu(t)+w(t)\bigr)\le g ,
\end{align*}
However, one-step safety does not guarantee recursive feasibility, even in the absence of disturbances. To illustrate this, consider the following system with $w(t)=0$ as
$x_1(t+1) = x_1(t) + x_2(t)$, $x_2(t+1) = x_2(t) + u(t)$
on $\mathcal{S} = \{x : |x_1| \le 1, \, |x_2| \le 1\}$ with input constraints
$\mathcal{U} = \{u : |u| \le 0.5\}$ and stage cost $\ell(x,u)=u^2$. 
Starting from $x(0)=(0.5,0.4)$,
the one-step optimization yields $u(0)=0$,
resulting in $x(1)=(0.9,0.4)\in\mathcal S$.
However, at the next step,
$x_1(2)=1.3>1$ regardless of $u(1)$,
making the optimization infeasible.
In contrast, choosing $u(0)=-0.4$ gives
$x(1)=(0.9,0)$, which remains in
$\mathcal S$ indefinitely under $u(t)=0$
for all $t\ge1$.
Thus, the failure is caused by the myopic
control decision rather than the impossibility
of maintaining safety from the initial state.
This example demonstrates that one-step safety
does not guarantee recursive feasibility,
motivating the synthesis of feedback controllers
that render an appropriate subset of the
state constraint set control invariant.
\end{remark}
}
\vspace{-9pt}

\subsection{Data-based Representation of Closed-loop Systems}
Inspired by \cite{Data4}, we derive a data-based closed-loop representation under the controller \eqref{contNew}. 
\begin{lem}
 Consider the system \eqref{system} with collected input--state data given
by \eqref{data-u}--\eqref{data-xx}. Let the
controller be given by \eqref{contNew}. Then, under Assumption 5, the data-based closed-loop representation of the system is given by
\begin{align}\label{data-f}
x(t+1) & = (X_1-W_0)G_{K,1} x(t) \nonumber \\
& + (X_1-W_0)G_{K,2}Q(x(t)) + w(t),
\end{align}
where 
\begin{align} \label{datacond}
     V_0 [G_{K,1} \, G_{K,2}]=I, \,
     K_1=U_0 G_{K,1}, \,  K_2=U_0 G_{K,2},
\end{align}
with $G_{K,1} \in \mathbb{R}^{T \times n}$
 and $G_{K,2} \in \mathbb{R}^{T \times N}$. Moreover, under Assumption 5, the solutions $G_{K,1}$  and $G_{K,2}$ exist and are not unique (and thus can be used as a decision variable).
\end{lem} \vspace{3pt}
\noindent \textit{Proof: }
Using the data \eqref{data-x}-\eqref{data-xx}, and the system \eqref{system}, one has
\(
 X_1  = \bar A_1 X_0+A_2 Q(X_0)+BU_0+W_0,
\)
where $\bar A_1$ is defined in \eqref{barA}. Multiplying both sides of this equation by
$ G_K = [G_{K,1} \quad G_{K,2}]$ and Using  \eqref{data-z}, one has
\(
 X_1 G_K= [\bar A_1 \quad A_2] V_0 G_K+ BU_0 G_K+ W_0 G_K,
\)
or equivalently,
\(
 (X_1-W_0) G_K=([\bar A_1 \quad A_2] V_0+BU_0) G_{K}.
\)
Using \eqref{datacond} in this equation, the data-based representation of the closed-loop system becomes $(X_1-W_0) G_K=A+BK$, or equivalently,
\begin{align}\label{system-cl2} 
& (X_1-W_0) G_{K,1}=\bar A_1+BK_1, \nonumber \\ & (X_1-W_0) G_{K,2}=A_2+BK_2.
\end{align}
By Assumption 5, a right inverse $G_{K}$ exists such that $V_0 \, G_{K}=I$. {Besides, since the rank of $V_0$ is $n+N$ while at least $n+N+1$ samples are collected, its right inverse $G_K$  exists and is not unique.} Using \eqref{system-cl2} in \eqref{clNew} gives \eqref{data-f}.  \hfill   $\blacksquare$ \vspace{6pt}


\section{Data-based Nonlinear Safe Control using Nonlinearity Minimization}


We adopt the nonlinearity-minimization strategy of \cite{Data4} as a baseline.
Many practical applications prescribe safe-set templates {a priori} from physical and operational bounds. Given such a template, Section~IV shows how to maximize the certified invariant set. We also show that nonlinearity cancellation is overly
conservative.
 \vspace{6pt}

The following theorem designs a safe controller via nonlinearity
cancellation and Lipschitz-bound minimization. \vspace{6pt}

\begin{thm}
\label{cor:disturbance_selfcontained}
Consider the system \eqref{system} under the control input
\eqref{contNew}. Let Assumptions~1--6 be satisfied.  Let there exist
$P_s\in\mathbb R^{s\times s}$, $\eta =[\eta_1,\ldots,\eta_s]^\top \in \mathbb{R}^s$, $G_{K,1} \in\mathbb R^{T\times n}$, and $G_{K,2} \in\mathbb R^{T\times N}$ that solve
\begin{subequations}
\begin{align} 
& \min_{G_{K,1},G_{K,2},\eta,P_s} \quad \|\eta\|_1 \label{eta}\\ 
& P_s g+\eta \le \lambda g , \label{cor:cont1}\\
& P_s F = F X_1 G_{K,1}, \label{cor:cont2}\\
& L_Q \, M_x \|F_{i,:}X_1G_{K,2}\|_2  + \nu_i \,
\big(\|G_{K,1}\|_2 + \|G_{K,2}\|_2\, L_Q \big) 
 \nonumber \\ & +h_w \, \|F_{i,:}\|_1  \le  \eta_i, \,\, \forall i \in \mathcal I_s,  \label{th1bound}           \\ & V_0\begin{bmatrix}G_{K,1}&G_{K,2}\end{bmatrix}=I, \label{cor:cont3}  \\
& P_s \ge 0, \, \, \eta \ge 0, \label{cor:cont4-1}
\end{align}
\end{subequations}
where $L_Q$ is the Lipschitz bound defined in Assumption 6, $\nu_i=T h_w \, M_x \|F_{i,:}\|_1$ and $M_x=\max_{x\in\cal S(F,g)} \|x\|_2$. Then, the  set {\eqref{eq:safeset}} is
$\lambda$-contractive with data-dependent control gain
$K_1=U_0G_{K,1}$ and $K_2=U_0G_{K,2}$.
\end{thm}
\noindent \textit{Proof:} See Appendix. \hfill $\blacksquare$ \vspace{3pt}

\begin{remark}
  {In the absence of disturbances, the constraint \eqref{th1bound}  reduces to
  \(
  L_Q \, M_x \|F_{i,:}X_1G_{K,2}\|_2  \le \eta_i.
  \)
The rest of the theorem remains unchanged.}
\end{remark} \vspace{3pt}

{The following example demonstrates that Theorem~1 may cancel beneficial nonlinear terms.} \vspace{3pt}

{\noindent\textbf{Motivating Example~1:}
Consider the scalar nonlinear system
\begin{equation}
x(t+1)=1.2x(t)-0.2x^3(t)+u(t),
\label{eq:motivating-example-1}
\end{equation}
with the prescribed safe set
$\mathcal S(F,g)=[-1,0]$, where
$F=[-1\ \ 1]^\top$ and $g=[1\ \ 0]^\top$.
Under $u(t)=k_1x(t)+{k_2}x^3(t)$, define
the closed-loop map
$f(x)=(1.2+k_1)x+({k_2}-0.2)x^3.$ For each facet $i$, define the closed-loop
facet map $H_i(x)=F_{i,:}f(x)$.
Since $F_{1,:}=-1$ and $F_{2,:}=1$,
we have $H_1(x)=-f(x)$ and
$H_2(x)=f(x)$. 
Theorem~1
forces cancellation of the nonlinear term
(${k_2}=0.2$) and safety is subsequently enforced through
$k_1$, since the uncanceled linear
dynamics are unstable.
However, the nonlinear term is beneficial
to invariance. To see this, consider
the uncontrolled system,
$u(t)\equiv0$, for which
$f(x)=1.2x-0.2x^3$.
Since $f'(x)=1.2-0.6x^2\ge0.6>0,
\quad x\in[-1,0]$, 
$f$ is increasing on $\mathcal S(F,g)$.
Moreover, $f(-1)=-1$ and $f(0)=0$.
Consequently, $-1\le f(x)\le 0,$ 
which proves that $\mathcal S(F,g)$ is
invariant without any control input. In contrast, without the cubic term ($-0.2x^3(t)$),
the uncontrolled linear dynamics give
$f(-1)=-1.2\notin\mathcal S(F,g)$. 
Furthermore, $f''(x)=-1.2x\ge0$ on
$\mathcal S(F,g)$, showing that the
nonlinearity has a favorable curvature
   structure {for the upper facet $x\le0$} that can be exploited in safety verification.
This example illustrates that
Lipschitz-based nonlinearity minimization
can cancel nonlinear terms that are
essential to preserving safety.
}

{We now provide another example for which nonlinear cancellation is not entirely possible.} \vspace{6pt}

{\noindent\textbf{Motivating Example 2:}
Consider the safe set $\cal S=\{x:0\le x_1 \le r_1, \, 0\le x_2 \le r_2 \} $ for the system
\begin{align}
\begin{bmatrix}x_1(t+1)\\ x_2(t+1)\end{bmatrix}
&=
A\begin{bmatrix}x_1(t)\\ x_2(t)\end{bmatrix}
+
{\begin{bmatrix}0\\ b\end{bmatrix}} u
+
\begin{bmatrix}
c_u\,x_2^{3}(t)\\ c_m\,x_1^{3}(t)
\end{bmatrix},
\label{eq:ex_onesided_sys}
\end{align}
where $A=[a_{ij}]$, and $b\neq 0$, $c_u > 0$ and
$u(t) = k_{11} x_1(t) + k_{12} x_2(t) + {k_{21}}\,x_1^{3}+ {k_{22}}\,x_2^{3}(t)$.
Theorem~1 chooses ${k_{21}}=-\frac{c_m}{b}$ and ${k_{22}}=0$ to cancel the nonlinear
term in the input channel. Here $F$ has four rows, giving four facet maps
$H_i(x)=F_{i,:}x(t{+}1)$: two for the upper facets $x_1\le r_1$, $x_2\le r_2$ and
two for the lower facets $-x_1\le 0$, $-x_2\le 0$. Each
lower-facet normal is the negative of its upper counterpart, so each pair reduces
to a max--min check of the same facet function. For the upper facets
($F_{1,:}=[\,1\;\;0\,]$, $F_{2,:}=[\,0\;\;1\,]$), $H_i(x)=x_i(t{+}1)$, we have
\begin{alignat}{2}
& H_1(x) = x_1(t+1) = a_{11}x_1+a_{12}x_2 + c_u x_2^3, \label{eq:H1}\\
& H_2(x) = x_2(t+1) = \bar a_{21}x_1+\bar a_{22} x_2 +  \bar a_n x_1^3, \label{eq:H2}
\end{alignat}
where $\bar a_{21}=a_{21}+bk_{11}$, $\bar a_{22}=a_{22}+bk_{12}$, and
$\bar a_n=c_m+b{k_{21}}=0$. Safety is then certified by
$\max_{x\in\cal S} H_1(x)\le \lambda r_1$, $\max_{x\in\cal S} H_2(x)\le \lambda r_2$
(upper facets) and $\min_{x\in\cal S} H_1(x)\ge 0$, $\min_{x\in\cal S} H_2(x)\ge 0$
(lower facets). Consider $H_1$. Since $c_u>0$ and $x_2\ge 0$ on $\cal S$,
$\frac{\partial^2 H_1}{\partial x_2^2}(x)=6c_u x_2 \ge 0,\, \forall x\in\cal S$,
so $H_1$ is convex and attains its maximum at a vertex of
$\mathcal I_{\ell}=\{(0,0),(r_1,0),(0,r_2),(r_1,r_2)\}$. Hence,
$\max_{x\in\cal S} H_1(x)\le \lambda r_1$ holds if
$\max\big\{0,\; a_{11}r_1,\; a_{12}r_2+c_u r_2^3,\; a_{11}r_1+a_{12}r_2+c_u r_2^3\big\}
\le \lambda r_1.$
Contrast this with the Lipschitz verification of Theorem~1. On $[0,r_2]$, $x_2^3$
has Lipschitz constant $L=3r_2^2$, giving for all $x\in\cal S$
\(
H_1(x)\le a_{11}|x_1| + a_{12}|x_2| + 3|c_u| \, r_2^2 \, |x_2|,
\)
so that, with $|x_1|\le r_1$ and $|x_2|\le r_2$, a sufficient condition is
\begin{align}
a_{11} r_1 + \bigl(a_{12}+3|c_u|r_2^2\bigr) r_2
\;\le\;
\lambda r_1 .
\label{Exlip}
\end{align}
Even in the worst case $a_{11},\,a_{12},\,c_u \ge 0$, where $H_1$ is nondecreasing
and $\max_{x\in\cal S} H_1(x)=H_1(r_1,r_2)$, the exact condition
\(
a_{11} r_1 + (a_{12} + c_u r_2^2)\, r_2 \le \lambda r_1
\)
is strictly less conservative than \eqref{Exlip}, and the gap widens when this
sign condition fails, since the absolute-value envelope discards sign and
monotonicity. The gap thus stems not from the nonlinearity itself but from a
global Lipschitz envelope that ignores the geometry of $\cal S$.

\vspace{3pt}

\begin{remark}
While Theorem~1 requires Lipschitz continuity, enforcing closed-loop facet map convexity yields less conservative conditions. Convexity enables exact contractivity verification via vertex maximization, avoiding the conservative worst-case scaling inherent in Lipschitz-based approach.
\end{remark}

\section{A Low-Complexity Data-based Nonlinear Safe Control Design}

{In this section, we present a convexification–verification framework that
preserves convexity by {adding a quadratic curvature-compensation term to each closed-loop facet map},
avoiding nonlinearity cancellation and worst-case bounding.}

\subsection{Systems with no Disturbances}
We first formalize the disturbance-free case. 

For the closed-loop system \eqref{data-f} with no disturbances and the safe set {\eqref{eq:safeset}}, define the facet map as
\begin{align} \label{Hf}
& H_i(x)
=F_{i,:} \, x(t{+}1)= 
F_{i,:} X_1G_{K,1}x+F_{i,:} X_1G_{K,2}Q(x) \nonumber \\ & \qquad \, \, \, =c_i^\top x+h_i(x).
\end{align}
where
\begin{align} \label{ci}
  c_i^\top=F_{i,:}X_1G_{K,1}, \quad h_i(x)=F_{i,:} X_1G_{K,2}Q(x).
\end{align}



\begin{thm}
\label{thm:hybrid_lagrangian_dc_matrix}
Consider the system \eqref{system} with $w(t)\equiv 0$ under
Assumptions~1 and~3--5.
Define $q_\ell=Q(x_{\ell})$ for ${\ell}\in\mathcal I_{\ell}$, where $Q(.)$ is defined in \eqref{Q}.
Let there exist decision variables
\(
\varepsilon\in\mathbb R^s,\;
G_{K,1},G_{K,2}
\)
that solve
\begin{subequations}
\begin{align}
& \min_{G_{K,1},G_{K,2},\varepsilon}
\quad \|\varepsilon\|_1
\label{eq:obj_cost} \\
& {F_{i,:}X_1G_{K,1}x_\ell+F_{i,:}X_1G_{K,2}q_\ell
+\tfrac12\,\varepsilon_i\|x_\ell\|_2^2 \;\le\; \lambda g_i,}
\nonumber \\ & \hspace{2.3cm}
{\forall i\in\mathcal I_s,\ \forall \ell\in\mathcal I_\ell,}
\label{eq:margin_merged_opt}\\
&  \frac{\partial^2 H_i(x)}{\partial x^2}
\succeq -\varepsilon_i I,
\quad \forall x\in\cal S(F,g),\ \forall i\in\mathcal I_s,
\label{eq:eps_hess_matrix}\\
& V_0\begin{bmatrix}G_{K,1}&G_{K,2}\end{bmatrix}=I,
\label{cor:cont4}\\
& \varepsilon\ge 0.
\label{eq:admissible_matrix}
\end{align}
\end{subequations}
Then, the resulting gains
\(K_1=U_0G_{K,1}\) and \(K_2=U_0G_{K,2}\)
render the closed-loop system $\lambda$-contractive on $\cal S(F,g)$.
\end{thm}
 \noindent \textit{Proof:} See Appendix. \hfill $\blacksquare$
\vspace{2pt}

{Since the linear and convexified nonlinear terms are evaluated jointly at each vertex in \eqref{eq:margin_merged_opt}, the conservatism introduced by bounding their maxima separately is avoided.}
\vspace{4pt}

\begin{corollary}
\label{cor:polyhedral_lyap_relaxed}
Let the conditions of Theorem~2 hold. Then,   \vspace{-3pt}
\begin{equation}\label{eq:V_def_cor}
V(x)=\max_{i\in \mathcal I_s}\frac{F_{i,:} \, x}{g_i}.
\end{equation}
is a Lyapunov function on $\cal S(F,g)$ for the closed-loop
system in Theorem~2.
\end{corollary}
\noindent \textit{Proof:} See Appendix. \hfill $\blacksquare$ \vspace{3pt}

\begin{remark}
For Motivating Example~1, $H_2(x)$ is  convex since
$\frac{\partial^2 H_2(x)}{\partial x^2}=-1.2\,x\ge 0$ for all $x\in[-1 \quad 0]$, and thus the Hessian condition in Theorem~2 holds with $\varepsilon{_2}=0$, and
$\max_{x\in\cal S} H_2(x)=\max\{H_2(-1),H_2(0)\}=0${. However, $H_1(x)=-f(x)$ is concave on $\cal S$, so \eqref{eq:eps_hess_matrix} requires $\varepsilon_1\ge 1.2$, and \eqref{eq:margin_merged_opt} at $x=-1$ gives $H_1(-1)+0.6=1.6\le\lambda$, which fails. Hence, Theorem~2 does not certify the invariance of $\cal S$ under $u(t)\equiv 0$, which follows in Example~1 from the monotonicity of $f$.} For Motivating Example~2, the closed-loop facet maps $H_i(x), i=1,2$ are convex on the safe
set $\cal S$, so the relaxed curvature condition in Theorem~2 again holds with
$\varepsilon{_i}=0$, and the verification reduces to exact vertex maximization,
$\max_{x\in\cal S}H_i(x)=\max_{\ell \in \mathcal I_\ell}H_i(x_\ell)$ {for the upper facets $i=1,2$, whereas the lower-facet map $-H_1$ is concave and requires $\varepsilon_3\ge 6c_ur_2$ (and $-H_2$ is concave whenever $\bar a_n>0$). In general, for a sign-symmetric pair $F_{j,:}=-F_{i,:}$, $\frac{\partial^2 H_j}{\partial x^2}=-\frac{\partial^2 H_i}{\partial x^2}$, so both facet maps are convex only if the nonlinear part along $F_{i,:}$ is cancelled. Therefore, in contrast to Theorem~1, Theorem~2 introduces no conservatism on the facets with favorable curvature and confines it, through $\varepsilon_i$, to the remaining facets.}
In general, however, curvature signs vary over $\cal S(F,g)$, so the conditions of
Theorem~2 cannot be satisfied globally with $\varepsilon=0$: in Example~1,
$ \frac{\partial^2 f(x)}{\partial x^2}=-1.2\,x$ changes sign on $[-r \quad r]$, so $f$ is not globally convex over
the symmetric set and curvature slack is unavoidable. When $\varepsilon >0$,
conservativeness arises from {the quadratic convexification} of the nonlinear facet map, which enforces
\(
 \frac{\partial^2 H_i(x)}{\partial x^2}
\succeq -\varepsilon_i I,\, x\in\cal S(F,g),
\)
so that $\tilde H_i(x)=H_i(x)+\tfrac{\varepsilon_i}{2}\|x\|_2^2$ becomes convex
and its maximum over $\cal S(F,g)$ is attained at a vertex.
{Theorem~2 and Proposition~1 below both certify safety without cancelling nonlinearities, but they treat $h_i$ differently. Theorem~2 convexifies all of $h_i$ through the slack $\tfrac12\varepsilon_i\|x\|_2^2$, which needs no Lipschitz bound but is charged at every vertex, and is tight when the curvature of $h_i$ has one sign on $\cal S(F,g)$. Proposition~1 splits $h_i$ into a part with nonnegative curvature, maximized exactly at the vertices with no slack, and a remainder bounded through Lipschitz constants. It is preferable when $h_i$ has strong curvature of mixed sign but a small Lipschitz constant. Proposition~1 never does worse than Theorem~1, whereas neither Theorem~2 nor Proposition~1 dominates the other.}
\end{remark}
\vspace{3pt}




\begin{prop}
{Consider the system \eqref{system} with $w(t)\equiv 0$ under Assumptions~1 and~3--5 and
the data-based controller parameterization \eqref{contNew}. 
Suppose the following convex optimization problem is feasible in the decision variables
\(
G_{K,1},\,G_{K,2},\,P \in \mathbb{R}^{s \times s},\,t\in\mathbb R^s,\;
\{z_i,r_i,k_i,a_i\}_{i\in\mathcal I_s}
\)
\begin{subequations}
\label{eq:hyb_opt}
\begin{align}
& \min \quad \sum_{i\in\mathcal I_s}\mathbf 1^\top k_i
\label{eq:hyb_obj}\\
& P g + t \le \lambda g, \quad P F = F X_1 G_{K,1}, \label{eq:hyb_margin}\\
& F_{i,:}X_1G_{K,2}=z_i^\top+r_i^\top,\qquad \forall i\in\mathcal I_s, \label{eq:hyb_split}\\
& \frac{\partial^2 (z_i^\top Q(x))}{\partial x^2}  \succeq 0,
\, \forall x\in\cal S(F,g),\  \forall i \in{\mathcal I_{s}} , \label{eq:hyb_convex}\\
& -k_i \le L_Q\odot r_i \le k_i,\qquad \forall i\in\mathcal I_s, \label{eq:hyb_abs_r}\\
& -a_i \le L_Q\odot (z_i+r_i) \le a_i,\qquad \forall i\in\mathcal I_s, \label{eq:hyb_abs_b}\\
& \mathbf 1^\top k_i \le \mathbf 1^\top a_i,\qquad \forall i\in\mathcal I_s, \label{eq:hyb_dominate}\\
&   z_i^\top q_\ell - P_{i,:}\!\big(g-Fx_\ell\big) + R_0\,\mathbf 1^\top k_i \le t_i,
\, \forall \ell\in\mathcal I_\ell,\ \forall i\in\mathcal I_s,
\label{eq:hyb_epi}
\\
& V_0 [G_{K,1}\ \ G_{K,2}] = I, \\ & P \ge 0, \quad k_i \ge 0, \quad a_i  \ge 0, \, \forall i\in\mathcal I_s,
\label{eq:hyb_data}
\end{align}
\end{subequations}
where $R_0=\max_{\ell\in\mathcal I_\ell}\|x_\ell\|_2$.
Then, $\cal S(F,g)$ is $\lambda$-contractive for the closed-loop system with gains $K_1=U_0 G_{K,1}$ and $K_2=U_0 G_{K,2}$. Moreover, 
the induced conservatism is never worse than that of Theorem~1.}
\end{prop}
 \noindent \textit{Proof:} See Appendix. \hfill $\blacksquare$
 \vspace{2pt}

{Every feasible point of \eqref{eq:hyb_opt}, not only its minimizer, certifies $\lambda$-contractivity: the objective \eqref{eq:hyb_obj} only selects, among these certificates, one with the smallest Lipschitz residual, and \eqref{eq:hyb_abs_b}--\eqref{eq:hyb_dominate} are used only for the comparison with Theorem~1.}
 \vspace{6pt}

{
{The next theorem instead makes the curvature compensation
{directional} and {local}.}}

 \vspace{3pt}
\vspace{3pt}

\begin{thm}
\label{thm:th3_face_restricted_relaxed}
{
Consider the system~\eqref{system} with $w(t)\equiv 0$ under Assumptions~1
and~3--5. Let
$\{x_\ell\}_{\ell\in{\mathcal J}}\subset\cal S(F,g)$ be a finite set of
nodes and let $\{\mathcal C_c\}_{c\in\mathcal I_c}$ be a finite family of
polytopes (cells) such that
\begin{align}
\cal S(F,g)=\bigcup_{c\in\mathcal I_c}\mathcal C_c,
\quad
\operatorname{vert}(\mathcal C_c)=\{x_\ell\}_{\ell\in{\mathcal J_c}},
\ \ {\mathcal J_c}\subseteq{\mathcal J} ,
\label{eq:th3_cells}
\end{align}
the default choice being the single cell $\mathcal C_1=\cal S(F,g)$, whose
nodes are the vertices of $\cal S(F,g)${, so that $\mathcal J=\mathcal I_\ell$}. Set $q_\ell=Q(x_\ell)$ and fix a
reference point $\bar x_c\in\mathbb R^n$ for every cell (e.g., its centroid).
Assume that there exist vectors $g_\ell\in\mathbb R^{T}$,
$\ell\in{\mathcal J}$, a matrix $G_{K,2}\in\mathbb R^{T\times N}$ and
symmetric matrices $E_{c,i}\in\mathbb R^{n\times n}$, $c\in\mathcal I_c$,
$i\in\mathcal I_s$, satisfying
\begin{subequations}
\label{eq:th3_face_opt_relaxed}
\begin{align}
& \min_{\{g_\ell\},\,G_{K,2},\,\{E_{c,i}\}}
\;\;\sum_{c\in\mathcal I_c}\sum_{i\in\mathcal I_s}
\operatorname{tr}\big(E_{c,i}M_c\big),
\nonumber\\
& \hspace{1.2cm}
M_c:=\!\!\sum_{\ell\in{\mathcal J_c}}\!\!
(x_\ell-\bar x_c)(x_\ell-\bar x_c)^\top
\label{eq:th3_obj_relaxed}\\
& F_{i,:}X_1g_\ell+F_{i,:}X_1G_{K,2}\,q_\ell
\nonumber\\
& \quad
+\tfrac12\,(x_\ell-\bar x_c)^\top E_{c,i}\,(x_\ell-\bar x_c)
\;\le\;\lambda g_i,
\nonumber\\
& \hspace{2.2cm}
\forall i\in\mathcal I_s,\ \forall c\in\mathcal I_c,\
\forall \ell\in{\mathcal J_c},
\label{eq:th3_budget_relaxed}\\
& \frac{\partial^2 h_i(x)}{\partial x^2}+E_{c,i}\;\succeq\;0,
\quad \forall x\in\mathcal C_c,\ \forall c\in\mathcal I_c,\
\forall i\in\mathcal I_s,
\label{eq:th3_soc_face_relaxed}\\
& V_0\,g_\ell=\begin{bmatrix}x_\ell\\ 0\end{bmatrix},
\ \ \forall \ell\in{\mathcal J},
\qquad
V_0\,G_{K,2}=\begin{bmatrix}0\\ I_N\end{bmatrix},
\label{eq:th3_data_relaxed}\\
& E_{c,i}\succeq 0,\quad \forall c\in\mathcal I_c,\ \forall i\in\mathcal I_s,
\label{eq:th3_adm_relaxed}
\end{align}
\end{subequations}
where
\begin{align}
h_i(x)= F_{i,:}X_1G_{K,2}\,Q(x).
\label{eq:th3_H_relaxed}
\end{align}
Define the node inputs $u_\ell=U_0g_\ell$ and the gain $K_2=U_0G_{K,2}$. For
any $x\in\cal S(F,g)$, select a cell $\mathcal C_{c}$ with $x\in\mathcal C_c$
and any weights $\theta(x)$ supported on the vertices of that cell,
\begin{align}
\theta_\ell(x)\ge0,
\quad
\sum_{\ell\in{\mathcal J_c}}\theta_\ell(x)=1,
\quad
x=\sum_{\ell\in{\mathcal J_c}}\theta_\ell(x)\,x_\ell,
\label{eq:th3_theta_relaxed}
\end{align}
and apply the interpolated controller
\begin{align}
u(x)=\sum_{\ell\in{\mathcal J_c}}\theta_\ell(x)\,u_\ell+K_2\,Q(x).
\label{eq:th3_policy_relaxed}
\end{align}
Then, $\cal S(F,g)$ is $\lambda$-contractive for the closed-loop system.}
\end{thm}
 \noindent \textit{Proof:} See Appendix. \hfill $\blacksquare$

\begin{remark}
\label{rem:matrix_slack}
{Each $g_\ell$ enters
\eqref{eq:th3_budget_relaxed} only through the inequalities written at its
own node $x_\ell$, so the node inputs $u_\ell=U_0g_\ell$ are independent
decision variables coupled only through the shared nonlinear gain, and they
do not collapse to a single gain. By \eqref{eq:th3_data_relaxed},
$X_1g_\ell=\bar A_1x_\ell+Bu_\ell$ is the image of node $x_\ell$ under the
affine part of the dynamics, so \eqref{eq:th3_budget_relaxed} evaluates the
linear and the nonlinear contributions jointly at the same node.
The slack $E_{c,i}$ need only dominate the
negative curvature of $h_i$ on the cell $\mathcal C_c$, and its price in
\eqref{eq:th3_budget_relaxed} is
$\tfrac12(x_\ell-\bar x_c)^\top E_{c,i}(x_\ell-\bar x_c)$, which scales with
the squared size of the cell rather than of $\cal S(F,g)$. Refining the cells
therefore reduces the curvature-induced tightening, while every facet remains
certified at every point of $\cal S(F,g)$. For a box, cutting each coordinate interval into $p$ pieces gives
$p^n$ cells with $(p+1)^n$ nodes, the $2^n$ orthant cells ($p=2$) being a
natural choice when the curvature changes sign at the origin. Only the
coordinates on which $Q$ depends need to be cut.
Since $\partial^2h_i/\partial x^2$ is linear in
$G_{K,2}$, \eqref{eq:th3_face_opt_relaxed} is a convex program. When each
$\partial^2Q_k/\partial x^2$ is affine in $x$,
\eqref{eq:th3_soc_face_relaxed} is an LMI affine in $x$ and it suffices to
impose it at the vertices of each cell, so \eqref{eq:th3_face_opt_relaxed} is
a finite semidefinite program. If the basis Hessians are in addition diagonal,
so may $E_{c,i}$ be taken, and the program reduces to a linear program.}
\end{remark}

\begin{remark}
{
Consider Motivating Example~2 with the constraint on $x_1$ replaced by the symmetric constraint $-r \le x_1 \le r$.
The two facet inequalities are then
$H_{1,+}(x)=x_1(t{+}1)\le \lambda r$ and $H_{1,-}(x)=-x_1(t{+}1)\le \lambda r$.
A cubic term has curvature proportional to its argument, so each signed facet
map is convex on one half of the interval and concave on the other. Under
Theorem~2, the slack must dominate the negative curvature over the whole
interval $[-r \quad r]$ and is charged over the full extent of the set at
every vertex. Under Theorem~\ref{thm:th3_face_restricted_relaxed} with the two
cells obtained by splitting the set at the origin, each signed facet map has
curvature of one sign on each cell: on the cell where it is convex no slack is
needed, and on the cell where it is concave the slack $E_{c,i}$ is charged
only along the curved direction and only over the half-width of the cell,
measured from $\bar x_c$. Both facets are certified on both cells, so the
conservatism due to curvature sign changes is reduced.}
\end{remark}

\subsection{Systems with Additive Disturbances}
We now develop a disturbance-aware version of Theorem~2 and discuss how it extends
Proposition~1 and Theorem~3. \vspace{3pt}


\begin{thm}
\label{thm:hybrid_dc_disturbance_in_convexity_relaxed}
Consider the system \eqref{system} under additive disturbances and let Assumptions~1--5 hold.
Define $q_\ell=Q(x_\ell)$ for all ${\ell} \in \mathcal I_{\ell}$.
Introduce decision variables
\(
\varepsilon\in\mathbb R^s,\;
G_{K,1},G_{K,2},
\)
{and epigraph variables
$\xi^{(1)},\xi^{(2)}\in\mathbb R_{\ge0}$.}
Consider the convex optimization
\begin{subequations}
\label{eq:robust_hybrid_dc_opt_relaxed}
\begin{alignat}{2}
& \min_{G_{K,1},G_{K,2},\varepsilon,\xi^{(1)},\xi^{(2)}}
 \| \varepsilon\|_1
\label{eq:robust_obj_relaxed}\\
& {F_{i,:}X_1G_{K,1}x_\ell+F_{i,:}X_1G_{K,2}q_\ell
+\tfrac12\,\varepsilon_i\|x_\ell\|_2^2+(d_T)_i \;\le\; \lambda g_i,}
\nonumber \\ & \hspace{2.3cm}
{\forall i\in\mathcal I_s,\ \forall \ell \in \mathcal I_{\ell},}
\label{eq:robust_budget_relaxed}\\&
 {\|G_{K,1}x_\ell\|_2 \le \xi^{(1)},\ \forall \ell \in \mathcal I_{\ell},
\, {\textstyle\sum_{k=1}^{N}\bar q_k\,\big\|(G_{K,2})_{:,k}\big\|_2 \le \xi^{(2)},}}
\label{eq:robust_epi_norms_relaxed}\\&
 \frac{\partial^2 H_i(x)}{\partial x^2}
\succeq -\varepsilon_i I,
\, \forall x\in\cal S(F,g),\;  \forall i \in\mathcal I_{s} ,
\label{eq:robust_hess_relaxed}\\&
 V_0 [G_{K,1}\;\;G_{K,2}] = I,\\&
{\varepsilon\ge 0, \quad \xi^{(1)} \ge 0,\quad \xi^{(2)} \ge0,}
\label{eq:robust_adm_relaxed}
\end{alignat}
\end{subequations}
{where $(d_T)_i=h_w \|F_{i,:}\|_1\big(\sqrt{T}\,(\xi^{(1)}+\xi^{(2)})+1\big)$ and $\bar q_k=\max_{x\in\cal S(F,g)}|Q_k(x)|$, $k=1,\dots,N$, which are computable a priori since $Q$ is known.}  
Then, the gains $K_1=U_0G_{K,1}$ and $K_2=U_0G_{K,2}$ returned by
\eqref{eq:robust_hybrid_dc_opt_relaxed} render the safe set
$\lambda$-contractive in the presence of disturbances.
\end{thm}
 \noindent \textit{Proof:} See Appendix. \hfill $\blacksquare$
\vspace{3pt}

\begin{remark}
{A purely additive robustification, as in Theorem~1, bounds the worst-case
disturbance effect independently of the safe-set geometry. In contrast,
Theorem~\ref{thm:hybrid_dc_disturbance_in_convexity_relaxed} embeds the disturbance
directly into the {convexified vertex conditions} through the {margin
$(d_T)_i$, which depends on the decision variables through $\xi^{(1)}$ and
$\xi^{(2)}$ and is} optimized jointly with $G_{K,1}$ and $G_{K,2}$. }
\end{remark}

\vspace{3pt}

{
\begin{remark}
\label{rem:th3_robust}
The disturbance-aware version of
Theorem~\ref{thm:th3_face_restricted_relaxed} is obtained with the same
margin. With disturbed data, \eqref{eq:th3_data_relaxed} gives
$X_1g_\ell=\bar A_1x_\ell+Bu_\ell+W_0g_\ell$ and
$X_1G_{K,2}=A_2+BK_2+W_0G_{K,2}$, so that under
\eqref{eq:th3_policy_relaxed}, for $x\in\mathcal C_c$,
\(
F_{i,:}x(t+1)=\sum_{\ell\in{\mathcal J_c}}\theta_\ell(x)F_{i,:}X_1g_\ell+h_i(x)
-F_{i,:}W_0\big(\sum_{\ell\in{\mathcal J_c}}\theta_\ell(x)g_\ell+G_{K,2}Q(x)\big)+F_{i,:}w.
\)
Since $\|\sum_\ell\theta_\ell(x)g_\ell\|_2\le\max_\ell\|g_\ell\|_2$, it
suffices to add the epigraph constraints
\(
\|g_\ell\|_2\le\xi^{(1)},\ \forall\ell\in{\mathcal J},
\,
\sum_{k=1}^{N}\bar q_k\big\|(G_{K,2})_{:,k}\big\|_2\le\xi^{(2)},
\)
and to replace the right-hand side $\lambda g_i$ of
\eqref{eq:th3_budget_relaxed} by $\lambda g_i-(d_T)_i$, with $(d_T)_i$ and
$\bar q_k$ as in
Theorem~\ref{thm:hybrid_dc_disturbance_in_convexity_relaxed}. The proof
combines \eqref{eq:th3_Htilde_relaxed}--\eqref{eq:th3_final} with
\eqref{eq:W0_support_relaxed}--\eqref{eq:G2Q_bound_relaxed}. The same
modification applies to Proposition~1 and 
Theorem~3.
\end{remark}
}%

\subsection{Discussion}
{
Theorems~1--3 and Proposition~1 form a menu of convex safety certificates with complementary tradeoffs.
Proposition~1 strictly improves upon Theorem~1 and can be progressively refined by integration with the vertex-based construction of Theorem~3 to further reduce conservatism.
In contrast, the conservatism of Proposition~1 relative to Theorem~2 depends on the structure of the nonlinearities. A priori analysis of the nonlinear basis functions (e.g., curvature or Lipschitz properties) can therefore guide the choice between geometry-Lipschitz certificates (Proposition~1 and its Theorem~3-based variants) and pure geometry- based certificates (Theorems~2 and~3) under a given computational budget. 
The only hyperparameter is the contractivity parameter, which trades off safe-set size and convergence rate.
}

\vspace{6pt}
 \subsubsection{Feasibility}
 {Infeasibility can arise either from fundamental limitations of the control problem (structural infeasibility) or from conservatism in the safety certificate. While the former cannot be remedied algorithmically, a key contribution of this paper is to address the latter.  If a theorem is infeasible for a prescribed polyhedral set due to structural or algorithmic limitations, feasibility can be recovered by adding a facet-wise scaling step.
Define \(r\in\mathbb{R}^s_{+}\) and the scaled set
\[
\cal S(r)=\{x\in\mathbb{R}^n:\;Fx\le\operatorname{diag}(r)\,g\}.
\]
Since the contractivity conditions are enforced facet-wise and depend linearly on \(g\), all results extend by replacing \(g\) with \(\operatorname{diag}(r)g\).
Thus, \(r\) can be optimized to identify a feasible invariant set while preserving the geometry of the original safe set and reducing conservatism across state directions.}

\vspace{6pt}
\subsubsection{Computational complexity}
{Reducing conservatism via Theorems~2--3, Proposition~1, and their
disturbance-aware counterparts incurs additional computational cost relative to
Theorem~1. {For Theorem~2 this stems from imposing the merged
vertex inequality \eqref{eq:margin_merged_opt} at every (facet, vertex) pair
together with the second-order conditions \eqref{eq:eps_hess_matrix}, i.e.\ from
vertex enumeration rather than from any dual certificate: no multipliers are
introduced, and the additional cost relative to Theorem~1 is $s|\mathcal I_\ell|$
linear inequalities plus $s$ curvature conditions. {Theorem~3 with the single cell
$\mathcal C_1=\cal S(F,g)$ keeps the same inequality count, replaces each
scalar slack $\varepsilon_i$ by a symmetric matrix $E_{1,i}$ and the gain
$G_{K,1}$ by the node vectors $g_\ell$. With a partition, it has
$s\sum_{c\in\mathcal I_c}|{\mathcal J_c}|$ linear inequalities
\eqref{eq:th3_budget_relaxed} and $s|\mathcal I_c|$ semidefinite conditions
\eqref{eq:th3_soc_face_relaxed} of size $n\times n$, which for diagonal basis
Hessians collapse to linear constraints (Remark~\ref{rem:matrix_slack}). The
number of cells is the tuning knob that trades computation for conservatism:
only the coordinates on which $Q$ depends need to be partitioned (two out of
four in Section~\ref{sec:sim_results}-B), and the cells may be refined only
where the vertex inequalities are violated.}} The proposed approaches thus
represent a moderate and tunable overhead in exchange for a substantial reduction
in conservatism, a trade-off well justified in safety-critical applications for
the following reasons:}

\smallskip
\noindent {\textbf{Facet- versus vertex-scaling.}
In the proposed framework, second-order (curvature) constraints are indexed by the {facets} of the safe set, i.e., by $i\in\mathcal I_s$, rather than by its vertices. Accordingly, the number of such constraints scales with the number of facets $s$ in Theorem~2 {and in} Proposition~1, {and is $s|\mathcal I_c|$ in Theorem~3, where the scalar slack is replaced by one matrix slack of size $n\times n$ per cell. When the basis Hessians are affine in $x$, each of these conditions is imposed at the vertices of the set (Theorem~2) or of the cell (Theorem~3).} 
For example, for the commonly used hyper-rectangular safe set
\(
\cal S=\{x\in\mathbb R^n \mid -r\le x \le r\},
\) the number of facets is $s=2n${, whereas the number of vertices is $2^n$, so the vertex inequalities, rather than the curvature conditions, dominate the cost of Theorems~2--4 for such sets in higher dimensions. Cutting each of $d\le n$ coordinates into $p$ pieces multiplies the number of cells by $p^{d}$}.}

\smallskip
\noindent {\textbf{Vertex constraints as linear inequalities.}
The vertex-based conditions appearing in Proposition 1 and Theorems~2--4 enter the synthesis as
linear (or affine) inequalities in the decision variables once the vertices are fixed. Consequently, even when $\ell$ is large, these constraints can often be handled efficiently using constraint-generation or active-set strategies, in which only violated vertex constraints are added iteratively.}

\smallskip

{
\noindent \textbf{Node inputs and cells:} Theorem~3 is a single convex program in the node vectors $g_\ell$, the shared gain $G_{K,2}$ and the cell-wise slacks $E_{c,i}$. Each $g_\ell$ enters only the inequalities written at its own node, and each $E_{c,i}$ only those of its own cell, so the constraint matrix is sparse, and the coupling between cells is confined to $G_{K,2}$ and to the nodes shared by adjacent cells. Online, the controller locates a cell containing the current state and computes the convex weights of the state with respect to the vertices of that cell, which requires only a small linear program or a closed-form expression (e.g., multilinear weights for box cells) \cite{convexinterpol, SetB}.
}%

\begin{remark}
Define the facet--wise curvature operator
\(
\mathcal{L}_i(x)
=
-\sum_{k=1}^{N}
\big(F_{i,:} X_1 {(G_{K,2})}_{:,k}\big)
\frac{\partial^2 Q_k(x)}{\partial x^2}
\;\in\; \mathbb{R}^{n \times n}.
\) The relaxed second--order condition \eqref{eq:eps_hess_matrix} in Theorem~2 can then be written as
$d^\top \mathcal{L}_i(x) d
\;\le\;
\varepsilon_i\, d^\top d,
\,
\forall d \in \mathbb{R}^n,\;
\forall x \in \cal{S}(F,g),$
or, equivalently, in matrix form,
\(
\mathcal{L}_i(x)
\;\preceq\;
\varepsilon_i I,
\,\,
\forall x \in \cal{S}(F,g),
\)
which bounds the curvature of the $i$--th closed--loop facet map uniformly
over the safe set. {{The matrix slack of
Theorem~\ref{thm:th3_face_restricted_relaxed} replaces this by
$\mathcal L_i(x)\preceq E_{c,i}$ for all $x\in\mathcal C_c$, with
$\mathcal L_i$ formed with the shared $G_{K,2}$; every reformulation below
applies verbatim on each cell with $\varepsilon_iI$ replaced by $E_{c,i}$ and
$\cal S(F,g)$ by $\mathcal C_c$, since $E_{c,i}$ enters linearly.}}
When nonlinear terms outside the input channels have uniformly bounded second derivatives, the relaxed second-order condition \eqref{eq:eps_hess_matrix} is state independent (i.e., $\mathcal L_i(x)=\mathcal L_i$) and admits an LMI formulation. Using the S--procedure~\cite{convex}, the second-order condition can be enforced
over the polyhedral direction set
$\cal T = \big\{ d\in\mathbb R^n \ \big|\ F_{i,:}d \le 0,\ \forall i\in\mathcal I(x)
\big\}$, where $\mathcal I(x)=\{\,i\mid F_{i,:}x=g_i\,\}$ denotes the index set of
facets active at the boundary point $x\in\partial\cal S(F,g)$ at which the directional condition is evaluated. This restriction is sufficient since only directions that preserve feasibility of the active facets can increase the facet map locally.
For polynomial nonlinear systems, the relaxed second--order condition can
alternatively be enforced via a sum--of--squares (SOS) \cite{SOS} constraint of the form
\(
{d^\top\!\big(\varepsilon_i I-\mathcal L_i(x)\big)d}
\ \text{is SOS},
\,\, \forall d \in \mathcal T,\ \forall x \in \cal S(F,g).
\)
Using the S--procedure, this condition can be equivalently converted into a
single SOS constraint over the joint variables $(x,d)$. In particular, as
shown in \cite{SOSsoc}, by introducing SOS multipliers associated with the
constraints defining $\mathcal T$ and $\cal S(F,g)$, the above condition
can be enforced through an SOS optimization problem in the augmented state
$(x,d)$.
\end{remark}

\begin{remark}
\label{rem:th3_input}
{
Theorem~3 can also account for input constraints. Let the admissible input set be the polytope
\(
\mathcal U
= \{u\in\mathbb R^{m} \mid f_{u,j}^\top u \le (g_u)_j,\ j=1,\dots,r_u\},
\)
and define $\psi_j(x)=f_{u,j}^\top U_0G_{K,2}\,Q(x)$, $j=1,\dots,r_u$. Under
the interpolated controller \eqref{eq:th3_policy_relaxed}, for $x$ in a cell
$\mathcal C_c$,
\(
f_{u,j}^\top u(x)
=\sum_{\ell\in{\mathcal J_c}}\theta_\ell(x)\,f_{u,j}^\top u_\ell+\psi_j(x).
\)
Introduce symmetric matrices $\Gamma_{c,j}\succeq0$ and add to
\eqref{eq:th3_face_opt_relaxed} the constraints
\begin{align}
& f_{u,j}^\top U_0g_\ell+f_{u,j}^\top U_0G_{K,2}\,q_\ell
\nonumber\\
& \quad
+\tfrac12(x_\ell-\bar x_c)^\top\Gamma_{c,j}(x_\ell-\bar x_c)\le(g_u)_j,
\nonumber\\
& \hspace{3.0cm}\forall j,\ \forall c\in\mathcal I_c,\ \forall\ell\in{\mathcal J_c},
\label{eq:th7_input_epi_face}\\
& \frac{\partial^2\psi_j(x)}{\partial x^2}+\Gamma_{c,j}\succeq0,
\quad\forall x\in\mathcal C_c,\ \forall c\in\mathcal I_c,\ \forall j ,
\label{eq:th7_input_soc_face}
\end{align}
which are convex in $(\{g_\ell\},G_{K,2},\{\Gamma_{c,j}\})$. Then
$\psi_j(x)+\tfrac12(x-\bar x_c)^\top\Gamma_{c,j}(x-\bar x_c)$ is convex on
$\mathcal C_c$, and the argument
\eqref{eq:th3_Htilde_relaxed}--\eqref{eq:th3_final} in the proof of
Theorem~\ref{thm:th3_face_restricted_relaxed}, applied to $\psi_j$ with
$\Gamma_{c,j}$ in place of $E_{c,i}$ and $f_{u,j}^\top u_\ell$ in place of
$F_{i,:}X_1g_\ell$, yields $f_{u,j}^\top u(x)\le(g_u)_j$ for all $j$. Hence,
the interpolated controller \eqref{eq:th3_policy_relaxed} satisfies the input
constraints everywhere on the safe set, and joint state and input constraints
are certified by the same program.}
\end{remark}

\section{Simulation Results}
\label{sec:sim_results}

{This section evaluates the proposed approaches and compares their performance with Lipschitz-based bounding approaches. All optimization problems are solved using \textsc{YALMIP} with the \textsc{MOSEK}
solver. Simulations were conducted on a Dell XPS~13 laptop with an Intel Core~i7-1250U CPU and 16~GB of RAM running Windows~11~Pro.}

{To thoroughly evaluate the performance of our proposed method, we consider two benchmark systems:
A) A three-Dimensional system is utilized primarily for graphical visualization of the results and direct comparison against existing methods, and  
B) a four-dimensional setup (A single-link arm with a soft joint) is used that serves as a realistic testbed to demonstrate the practical viability, scalability, and robustness of our control framework under actual operational constraints.}
\subsection{Three-dimensional System}
{\paragraph{System Description} We consider the DT-NS
\begin{equation}
x(t+1)=A_1x(t)+A_2 Q(x(t))+Bu(t)+w(t),
\label{eq:sim_system}
\end{equation}
where $x(t)\in\mathbb{R}^3$ and $u(t)\in\mathbb{R}$, with
\[
A_1=
\begin{bmatrix}
0.90 & 0.02 & 0 \\
-0.3 & 0.85 & 0.01 \\
0.05 & 0 & 0.80
\end{bmatrix},
\quad
B=
\begin{bmatrix}
0 \\ 0.1 \\ 0
\end{bmatrix},
\]
and nonlinear basis functions and their coefficient matrix as
\[
Q(x)=
\begin{bmatrix}
x_1^3 & x_2^3 & x_3^3 & x_1^2
\end{bmatrix}^\top,
\]
\[
A_2 \;=\;
\begin{bmatrix}
e_1 & 0      & 0      & 0 \\
0      & -0.2  & 0      & 0 \\
0      & -0.008 & e_2 & -0.05
\end{bmatrix}.
\]
Throughout all experiments, safety is certified with respect to the box-shaped invariant set
\[
\mathcal S(F,g)=\{x\in\mathbb{R}^3 \mid |x_i|\le r,\ i=1,2,3\},
\]
and performance is evaluated both in terms of the admissible nonlinearity level and the maximum certifiable radius $r$ for fixed nonlinear coefficients.
To assess nonlinearity tolerance, we conduct two parametric studies: an \textit{$e_1$-sweep} with $e_2=-0.05$ fixed and an \textit{$e_2$-sweep} with $e_1=-0.01$ fixed, each identifying the largest admissible coefficient for which feasibility holds at $r=0.5$. To evaluate the size of the certified safe set, we perform an \textit{$r$-maximization} study by fixing $e_1=-0.01$ and $e_2=-0.005$ and maximizing the invariant set radius $r$.}

{\paragraph{Data Collection}
Controller synthesis is based on a single open-loop experiment of length \(T=20\), with input sequence
\(U_0=[u(0)\;\cdots\;u(19)]\) applied to the unknown system.
No knowledge of the system matrices \((A_1,B,A_2)\) is used during controller design; the controller is synthesized solely from measured input–state data.
The resulting state measurements \(x(t)\in\mathbb{R}^3\) are collected into
\(X_0=[x(0)\;\cdots\;x(19)]\) and \(X_1=[x(1)\;\cdots\;x(20)]\).
The data satisfy the rank condition in Assumption~5, which is verified numerically in all simulations.
The matrices \((A_1,B,A_2)\) are used only for data generation and closed-loop simulation, which is standard practice in data-driven control studies.}

\paragraph{Lipschitz-Based Formulation (Theorem~1)}
{Using Theorem 1 with no disturbance, for the $e_1$-sweep, $e_2$-sweep, and $r$-maximization experiments, we obtain $e_1=-0.04$,  $e_2=-0.015$, and $r=0.605$, respectively.
The corresponding control gains learned for the $r$-maximization scenario are \((K_1,K_2)=\begin{bmatrix}0.416 & -1.35 & 0.035 & 0 & 0.2 & 0 & 0\end{bmatrix}\).
The associated optimization problem for a specific value of $r$ is solved in $0.1457$~s. Figure~\ref{fig:thm1r} illustrates the closed-loop evolution under Theorem~1 (Lipschitz-only) certificate.
Trajectories are simulated from the safe-set vertices and show how the certified controller drives the state toward the interior while maintaining safety in simulation.}

\begin{figure}
    \centering
    \vspace{-0.9\baselineskip}
    \includegraphics[width=1\linewidth]{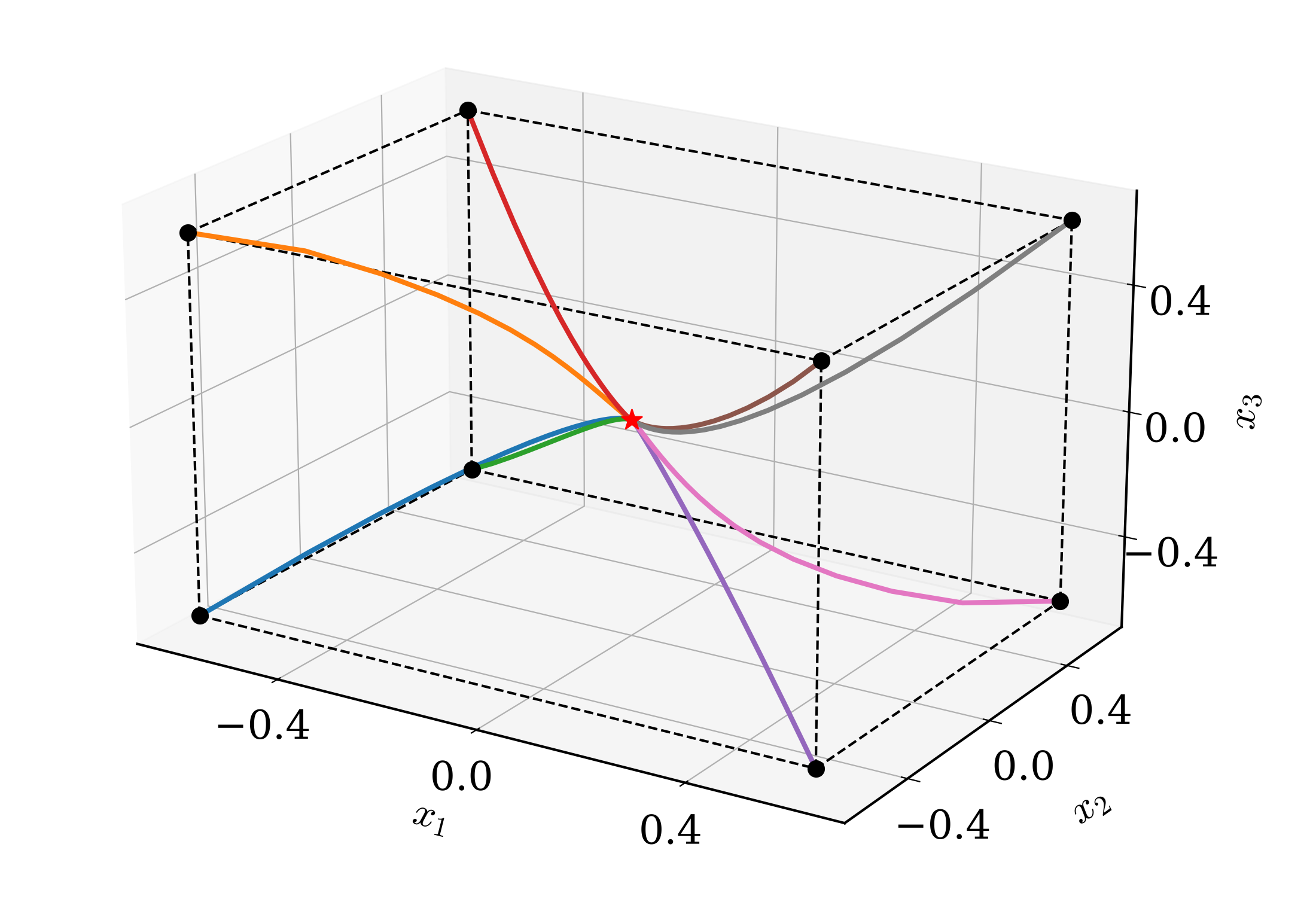} 
    \caption{Theorem~1: closed-loop state trajectories in $(x_1,x_2,x_3)$ from box-vertex initial conditions.}
    \label{fig:thm1r}
    \vspace{-0.8\baselineskip}
\end{figure}


\paragraph{Results of Proposed Approaches}

{%
Theorem~2 and Theorem~\ref{thm:th3_face_restricted_relaxed} are evaluated on
the same three studies with $\lambda=1$. All basis Hessians are diagonal and
affine in $x$, so the slacks are taken diagonal, the second-order conditions
are imposed at the vertices of each cell, and every program is a linear
program (Remark~\ref{rem:matrix_slack}). Theorem~3 is run on three partitions
of the box obtained by cutting every coordinate interval $[-r,r]$ into $p$
equal pieces: the single cell $\mathcal C_1=\cal S(F,g)$ ($p=1$, $8$ nodes),
the $8$ orthant cells ($p=2$, $27$ nodes), and $64$ cells ($p=4$, $125$
nodes), with $\bar x_c$ the cell centres and the multilinear interpolation
weights of each cell in \eqref{eq:th3_policy_relaxed}. The results are summarized in
Table~\ref{tab:sim3d}.}

\begin{table}[!t]
\centering
\caption{{Three-dimensional example:
largest certified $|e_1|$ and $|e_2|$ at $r=0.5$ and largest certified $r$.}}
\label{tab:sim3d}

\begin{tabular}{lccc}
\hline
 & $|e_1|$ & $|e_2|$ & $r$ \\ \hline
Theorem~1 & 0.04 & 0.015 & 0.605 \\
Theorem~2 & 0.040 & 0.061 & 1.00 \\
Theorem~3, one cell & 0.159 & 0.234 & 1.39 \\
Theorem~3, 8 cells & 5.23 & 4.92 & 2.29 \\
Theorem~3, 64 cells & 7.52 & 6.89 & 2.51 \\ \hline
\end{tabular}
\end{table}

{%
Theorem~2 matches Theorem~1 in the $e_1$-sweep and improves on it in the
$e_2$-sweep and in the certified radius, but
the scalar $\varepsilon_i$ needed to dominate the negative curvature on the
whole set is charged $\tfrac12\varepsilon_i\|x_\ell\|_2^2$ at every vertex.
With a single cell, the matrix slack of Theorem~3 charges this curvature only
along the coordinate in which the facet map actually curves, which gives a
fourfold larger $|e_1|$ and $|e_2|$ and a $1.4\times$ larger radius than
Theorem~2. The decisive improvement comes from the cells. On an orthant cell
every cubic term has curvature of one sign, so no slack is required for the
facet maps that are convex there, and for the remaining ones the penalty
$\tfrac12(x_\ell-\bar x_c)^\top E_{c,i}(x_\ell-\bar x_c)$ is measured from the
cell centre over half the width of the box. The certified nonlinearity
magnitudes then increase by more than an order of magnitude, and refining to
$64$ cells increases them further. Measuring the penalty from the cell centres
is essential.

Relative to Theorem~1, Theorem~3 with $64$ cells increases the tolerable
nonlinearity magnitude by $188\times$ ($|e_1|: 0.04 \rightarrow 7.52$) and
approximately $459\times$ ($|e_2|: 0.015 \rightarrow 6.89$), and enlarges the
certified invariant set by $4.1\times$ ($r: 0.605 \rightarrow 2.51$). When
the input constraint $|u|\le 1$ is imposed through
Remark~\ref{rem:th3_input}, the cell-wise versions retain the same
nonlinearity tolerance and certify $r=2.21$ ($8$ cells) and $r=2.51$ ($64$
cells); the trajectories reported below are obtained in this
input-constrained setting. As an independent
check, each closed loop was evaluated one step ahead on a $61^3$ grid and
$3\times10^5$ random points of the box, and no point left
$\lambda\cal S(F,g)$.}

\begin{figure}
    \centering
    \vspace{-0.1\baselineskip}
    \setlength{\abovecaptionskip}{2pt}
    \includegraphics[width=\linewidth]{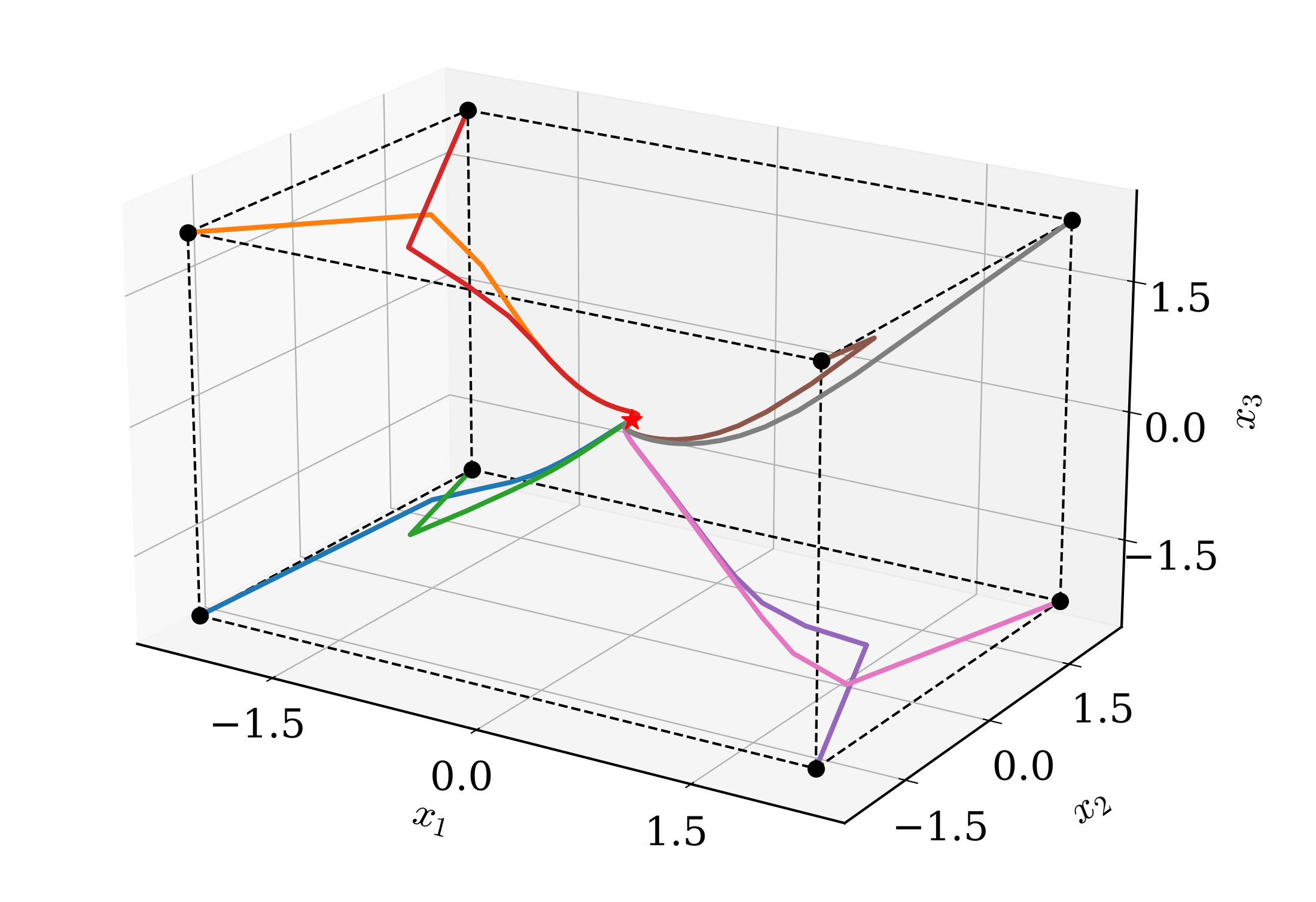}
    \caption{{Theorem~3 with $8$ orthant cells and $|u|\le1$ ($r=2.21$): closed-loop trajectories in $(x_1,x_2,x_3)$ from box-vertex initial conditions.}}
    \label{fig:thm3unstructr}
    \vspace{-0.8\baselineskip}
\end{figure}

\begin{figure}
    \centering
    \vspace{-0.1\baselineskip}
    \setlength{\abovecaptionskip}{2pt}
    \includegraphics[width=\linewidth]{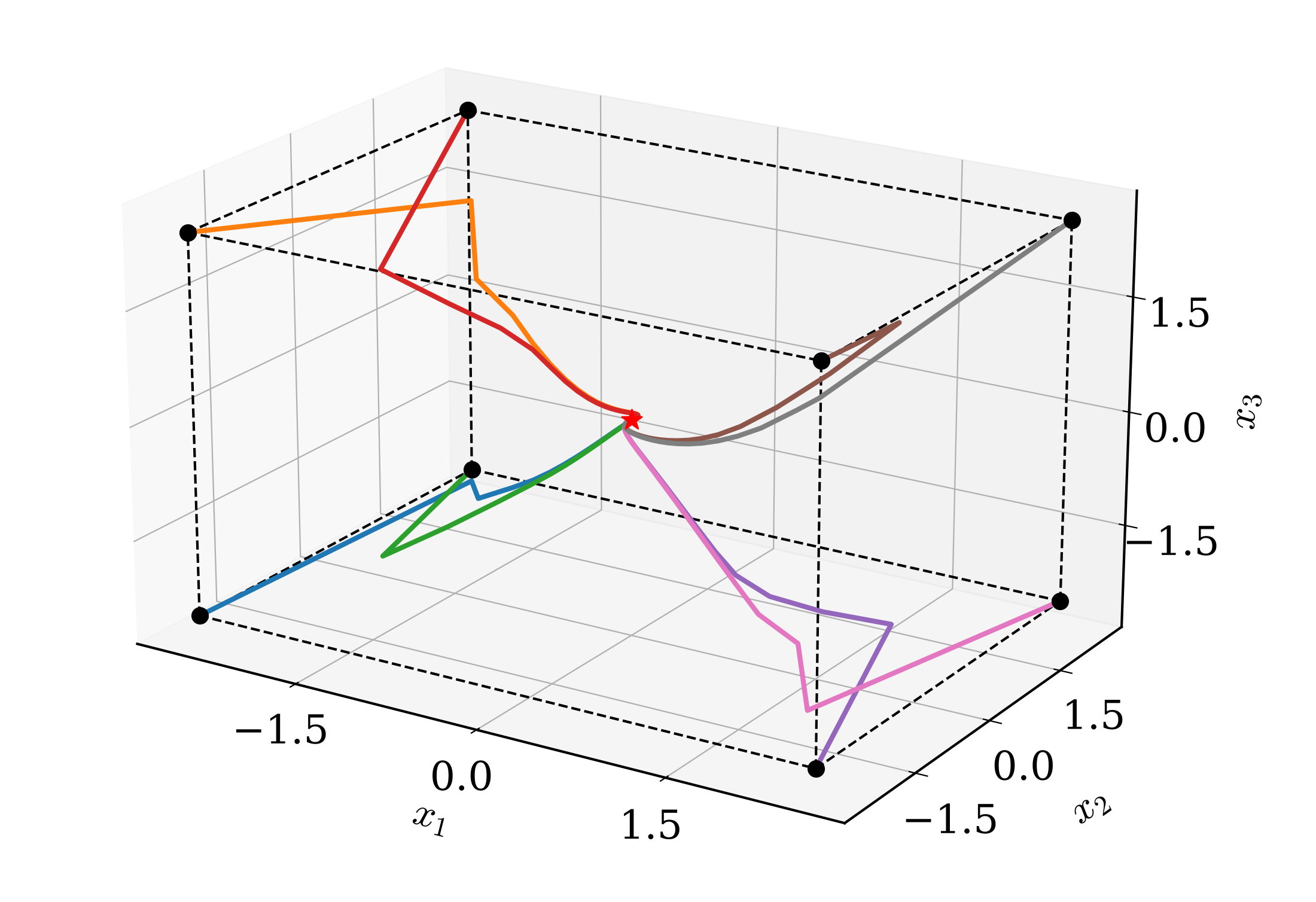}
    \caption{{Theorem~3 with $64$ cells and $|u|\le1$ ($r=2.50$): closed-loop trajectories in $(x_1,x_2,x_3)$ from box-vertex initial conditions.}}
    \label{fig:thm3structr}
    \vspace{-1.8\baselineskip}
\end{figure} 

\begin{figure}
    \centering
    \vspace{-0.07\baselineskip}
   \setlength{\abovecaptionskip}{2pt}
   \includegraphics[width=\linewidth]{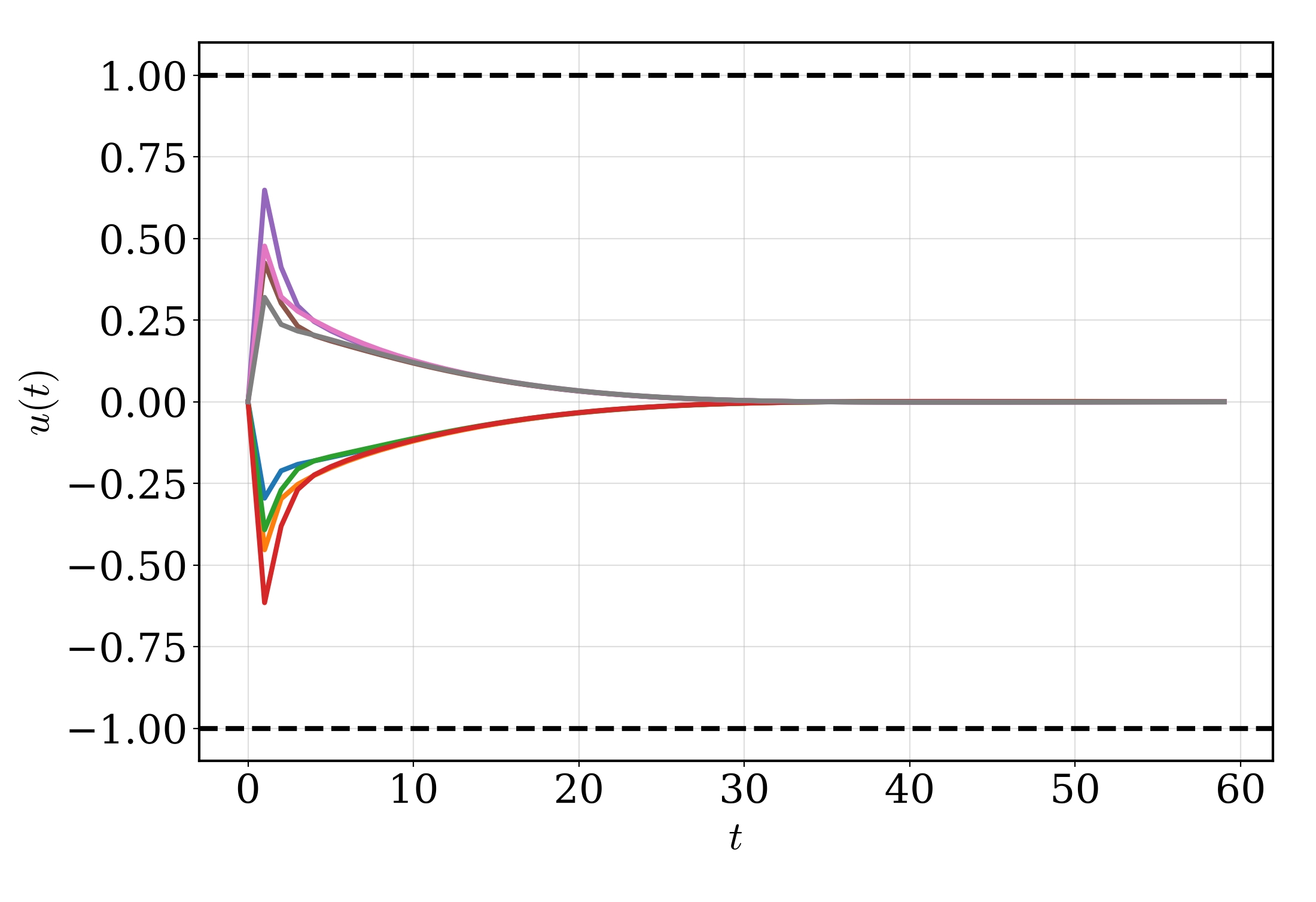}
    \caption{{Control input trajectories $u(t)$ for the vertex-initialized simulations of Fig.~\ref{fig:thm3unstructr}. Dashed: input bound $\pm1$.}}
    \label{fig:inputr}
    \vspace{-1.8\baselineskip}
\end{figure}

{%
Figures~\ref{fig:thm3unstructr} and~\ref{fig:thm3structr} report the
trajectories obtained in the $r$-maximization scenario with $8$ and $64$
cells, respectively, and Fig.~\ref{fig:inputr} shows the control effort for
the runs of Fig.~\ref{fig:thm3unstructr}. As can be seen, states
and inputs do not violate their constraints. }

\begin{figure}
    \centering
    \vspace{-0.1\baselineskip}
    \setlength{\abovecaptionskip}{2pt}
    \includegraphics[width=\linewidth]{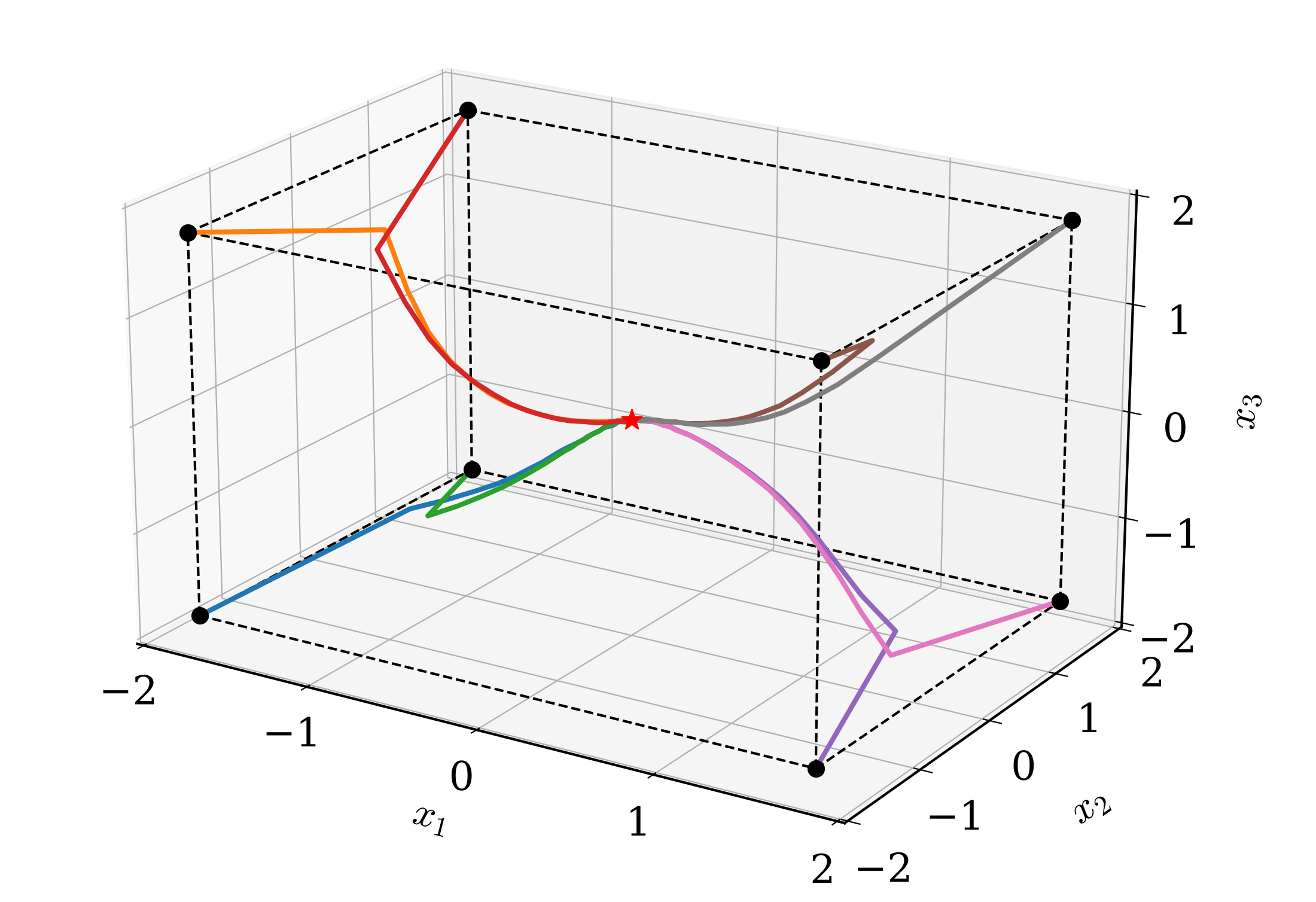}
    \caption{{Disturbance-aware version of Theorem~3 (Remark~\ref{rem:th3_robust}, $8$ cells, $h_w=5\times10^{-3}$, $r=1.78$): closed-loop trajectories in $(x_1,x_2,x_3)$ from box-vertex initial conditions.}}
    \label{fig:thm4}
    \vspace{-0.8\baselineskip}
\end{figure}

{%
\indent Finally, we introduce an additive disturbance to evaluate Theorem~4
and the disturbance-aware version of Theorem~3 in
Remark~\ref{rem:th3_robust}. In this case the certificates depend on the data
through $\xi^{(1)}$ and $\xi^{(2)}$, i.e. through the norms of the
data-consistent solutions, so the data set matters: we use $T=20$
input--state samples $(x(t),u(t),x(t+1))$ with states drawn uniformly from
$[-3,3]^3$ and inputs from $[-10,10]$, each successor state being corrupted by
a disturbance drawn uniformly from $\mathcal W$. For $h_w=5\times10^{-3}$,
Theorem~4 certifies $|e_1|=0.039$, $|e_2|=0.080$ and $r=0.97$; the
disturbance-aware version of Theorem~3 certifies $|e_1|=0.104$,
$|e_2|=0.175$ and $r=1.24$ with a single cell, and $|e_1|=5.0$, $|e_2|=4.77$
and $r=1.78$ with the $8$ orthant cells. For $h_w=10^{-2}$, Theorem~4 is
infeasible in the $e_1$-sweep and certifies $|e_2|=0.040$ and $r=0.74$,
whereas the $8$-cell version certifies $|e_1|=4.88$, $|e_2|=4.61$ and
$r=1.37$. For $h_w=3\times10^{-2}$ none of the programs is feasible at these
radii: the first state is not actuated, its row leaves a contraction budget
of only $(\lambda-0.92)\,r$, and the margin $(d_T)_i\ge h_w$ alone exceeds it
at $r=0.5$. In every feasible case the worst-case one-step successor,
including the worst disturbance, was verified to remain in
$\lambda\cal S(F,g)$ on the same grid and random samples as above.}
 All algorithms exhibit sensitivity to the disturbance magnitude, because although the disturbance sequence $W_0$ is realized but unknown, the controllers are designed using a worst-case bound holding uniformly over all realizations, so performance degrades as the disturbance grows. Recent results
in~\cite{unify} can be leveraged to more accurately approximate the realized disturbance sequence,
thereby tightening the uncertainty bound and reducing conservatism. Figure~\ref{fig:thm4}
illustrates the results obtained {with the $8$-cell certificate}, where the learned {node inputs and nonlinear gain}
are applied to the {disturbed} system and trajectories are initialized at the vertices of the safe set.

%

\subsection{Single-Link Arm with a Soft (Viscoelastic) Joint}
{
We demonstrate the proposed approach on a single-link manipulator mounted on
a soft, viscoelastic joint. The system has four states and exhibits two
distinct nonlinear effects: link-side gravity and the rate-dependent
structural damping characteristic of elastomeric and series-elastic
transmissions.  {The folliwng system is the Euler discretization, with sampling time $T_s=0.01$\,s, of the single-link flexible-joint robot model of \cite{Spong1987}, augmented with a linear link-side damping term ($0.996=1-0.4\,T_s$) and the cubic term $-0.05\,x_4^3$}.
\begin{align}
\label{eq:sj-plant}
x_1(t+1) &= x_1 + 0.01\,x_2, \nonumber \\
x_2(t+1) &= -0.486\,x_1 + 0.9875\,x_2 + 0.486\,x_3 + 0.216\,u,\nonumber \\
x_3(t+1) &= x_3 + 0.01\,x_4,\nonumber \\
x_4(t+1) &= 0.195\,x_1 - 0.195\,x_3 + 0.996\,x_4 \nonumber \\
         &\qquad - 0.0333\sin(x_3) - 0.05\,x_4^{3},
\end{align}
where $x_1$ is the motor angle (rad), $x_2$ the motor velocity (rad/s),
$x_3$ the link angle (rad), $x_4$ the link velocity (rad/s), and
$u\in\mathbb{R}$ the motor torque. The cubic term $-d\,T_s x_4^3$ with
$d=5$ models the nonlinear structural damping of the joint.
 The dictionary is
$S(x)=[\,\sin(x_3),\;x_4^3\,]^\top$ with
$A_s=\bigl[\begin{smallmatrix}0&0&1&0\\0&0&0&0\end{smallmatrix}\bigr]$,
so that the nonlinearity remainder is
$Q(x)=[\,\sin(x_3)-x_3,\;x_4^3\,]^\top$, and
$\bar A_1=A_1+A_2A_s$ absorbs the linear part of gravity into the
$(4,3)$ entry. 

The physical constraint set is the hyper-rectangle
$\mathcal X=\{x:\ |x_i|\le r_i\}$ with
$(r_1,r_2,r_3,r_4)=(0.18,\,1.40,\,0.072,\,0.25)$, i.e. $\pm10.3^\circ$ of
motor travel, $\pm80^\circ/$s of motor speed, $\pm4.1^\circ$ of link travel
and $\pm14.3^\circ/$s of link speed.
The set $\mathcal X$ is not control invariant, and no controller can make it
so. The kinematic rows satisfy $x_1(t+1)=x_1+T_sx_2$ and
$x_3(t+1)=x_3+T_sx_4$ with $F_{i,:}B=0$ on the corresponding facets, so at
the vertex $(r_1,r_2,\cdot,\cdot)$ the successor is
$r_1+T_sr_2>r_1\ge\lambda r_1$ regardless of $u$.  We therefore prescribe as template the maximal $\lambda_t$-contractive subset of
$\mathcal X$ for a nominal closed loop $\bar A_1+BK_{\mathrm{nom}}$, obtained
by the standard backward recursion
$F\leftarrow[F;\,F(\bar A_1+BK_{\mathrm{nom}})]$,
$g\leftarrow[g;\,\lambda_t g]$ initialised at $\mathcal X$ and iterated to
convergence with redundant inequalities removed. With $\lambda_t=0.985$ and
$K_{\mathrm{nom}}$ placing the nominal poles at $\{0.70,\,0.78,\,0.86,\,0.92\}$ this terminates after $20$ iterations. The nominal gain fixes the template geometry only, and the certificate is learned from data with $(\bar A_1,A_2,B)$ unknown. The certified set is then
reported through the scaling $\mathcal S(\alpha)=\{x:\ Fx\le\alpha g\}$ with
$\alpha$ maximized by bisection and $\lambda=0.99$ throughout.

Controller synthesis uses a single open-loop experiment of length $T=80$ with
$u(t)$ drawn uniformly from $[-30,30]$\,N$\cdot$m. The dictionary is
rescaled as $\tilde Q=Q/\sigma$ with $\sigma=\|Q(X_0)\|_F/\sqrt{T}$ before
forming $V_0$; this is a change of units on the basis, absorbed by $K_2$,
and leaves all conditions of Section~IV unchanged.
{%
The certificate is evaluated with
Theorem~\ref{thm:th3_face_restricted_relaxed}. The template has $s=68$ facets
and $402$ vertices. Since $Q$ depends on $x_3$ and $x_4$ only, the scaled
template $\mathcal S(\alpha)$ is partitioned by hyperplanes orthogonal to
these two coordinates into four slabs each, which gives $16$ cells; the nodes
are the cell vertices, $\bar x_c$ is the mean of the vertices of
$\mathcal C_c$, and the node vectors are parameterized by a single gain,
$g_\ell=G_{K,1}x_\ell$. The Hessians of the two basis functions are
$-\sin(x_3)\,e_3e_3^\top$ and $6x_4\,e_4e_4^\top$, so $E_{c,i}$ is taken
diagonal and supported on the $(x_3,x_4)$ block; both entries are monotone on
every cell, so \eqref{eq:th3_soc_face_relaxed} reduces to four scalar
inequalities per cell and facet, evaluated at the extreme values of $x_3$ and
$x_4$ on the cell, and \eqref{eq:th3_face_opt_relaxed} is a linear program.
All $68$ facets are certified at every point of $\mathcal S(\alpha)$, so the
certificate is deterministic and no a posteriori verification is used. The
nonlinear term enters through its {exact} values $q_\ell=Q(x_\ell)$ at the
nodes, and the convexification penalty is charged only along $(x_3,x_4)$ and
only over the extent of a cell, which is why a beneficial nonlinearity reduces
the required budget here while the envelope of Theorem~1 inflates it.

The bisection returns $\alpha=0.8914$, i.e. a certified
operating cell of $\pm8.79^\circ$ motor angle, $\pm71.5^\circ/$s motor
speed, $\pm3.57^\circ$ link angle and $\pm12.8^\circ/$s link speed
(extents $[0.15,\;1.25,\;0.06,\;0.22]$ in rad and rad/s). The effect of the
partition is visible on the same template: Theorem~2 certifies
$\alpha=0.089$, Theorem~\ref{thm:th3_face_restricted_relaxed} with the single
cell $\mathcal C_1=\mathcal S(\alpha)$ certifies $\alpha=0.445$, with $2\times2$
cells $\alpha=0.791$, and with $4\times4$ cells $\alpha=0.891$; refining to
$8\times8$ cells leaves $\alpha$ unchanged. The learned gains are
\(
 K_1=\begin{bmatrix}-82.97 & -3.34 & -75.98 & -41.74\end{bmatrix}, \,
K_2=\begin{bmatrix}0 & 0\end{bmatrix},
\)
with $\rho(\bar A_1+BK_1)=0.917$. The link-row entries are unchanged from
$A_2$, since $F_{i,:}B=0$ there and the plant's gravity and damping cannot be
cancelled, and the program returns $K_2=0$: the nonlinearities are left in the
loop and certified through their exact node values rather than cancelled or
reshaped.
Under the same template and the same data, Theorem~1 certifies only
$\alpha=0.1668$, so Theorem~3 certifies a $5.34\times$ larger safe set.}
The mechanism is the envelope: Theorem~1 charges $L_QM_x=0.190$ per unit
$|F_{i,:}(A_2+BK_2)|$, whereas the exact value at the vertices is
$\max_\ell|q_\ell|=1.11\times10^{-2}$, a $17\times$ inflation of the
nonlinear contribution.
}

{Figs.~\ref{fig:sj_x1}--\ref{fig:sj_x4} show the closed-loop
response for $8$ initial conditions placed at vertices of the certified set,
i.e.\ on its boundary, which is the worst case for invariance. Every state stays within its certified bound at all times and converges to the origin, so the synthesized controller is simultaneously safe and stabilizing.} {The figures are regenerated with the gains of Theorem~\ref{thm:th3_face_restricted_relaxed} reported above.}

\begin{figure}[!t]
  \centering
  \includegraphics[width=\linewidth]{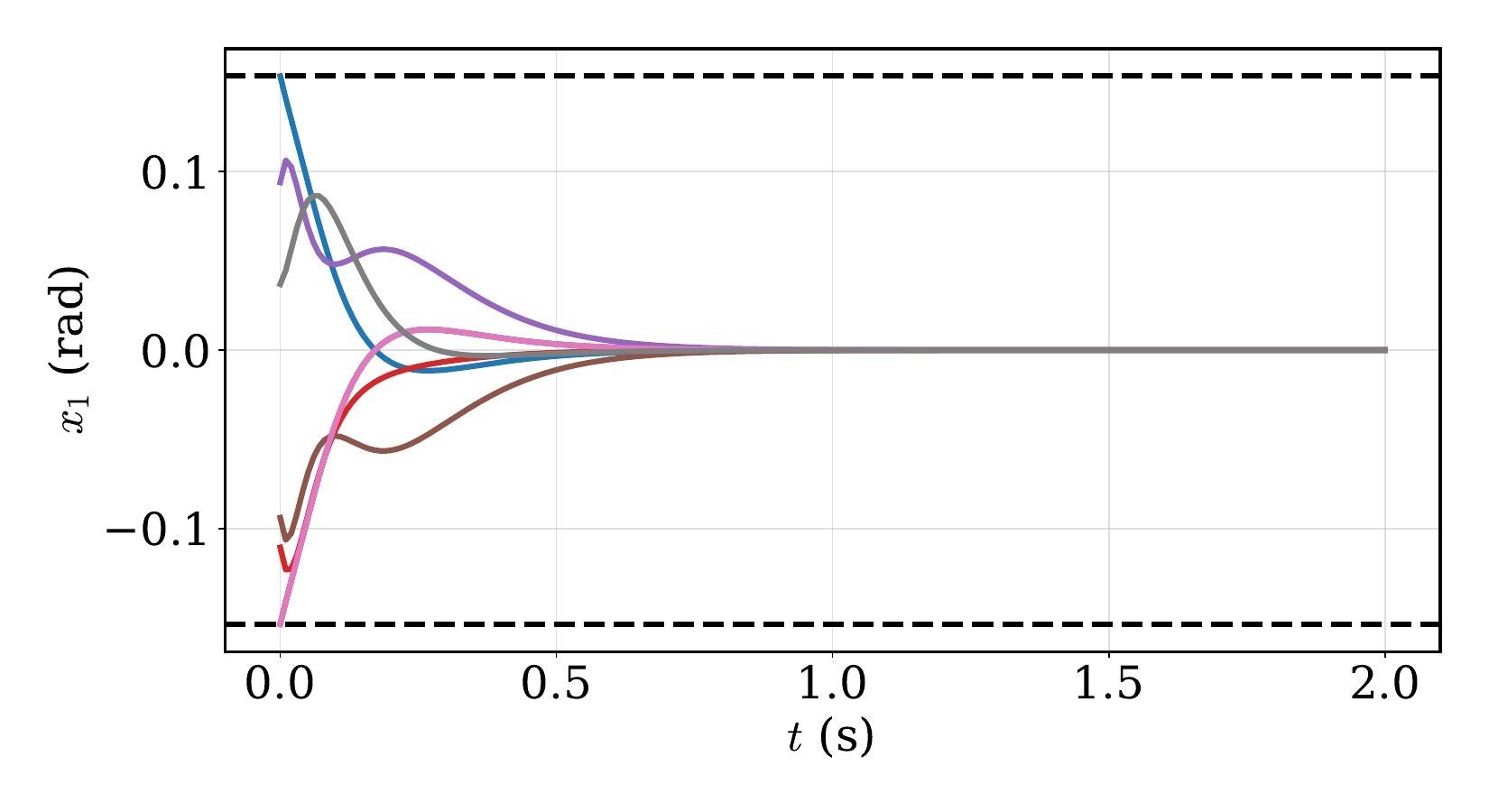}
  \caption{{Motor angle $x_1$ for $8$ initial conditions on
  $\partial\cal S(F,g)$. Dashed: certified bound $\pm0.1535$~rad.}}
  \label{fig:sj_x1}
  \vspace{-0.6\baselineskip}
\end{figure}

\begin{figure}[!t]
  \centering
  \includegraphics[width=\linewidth]{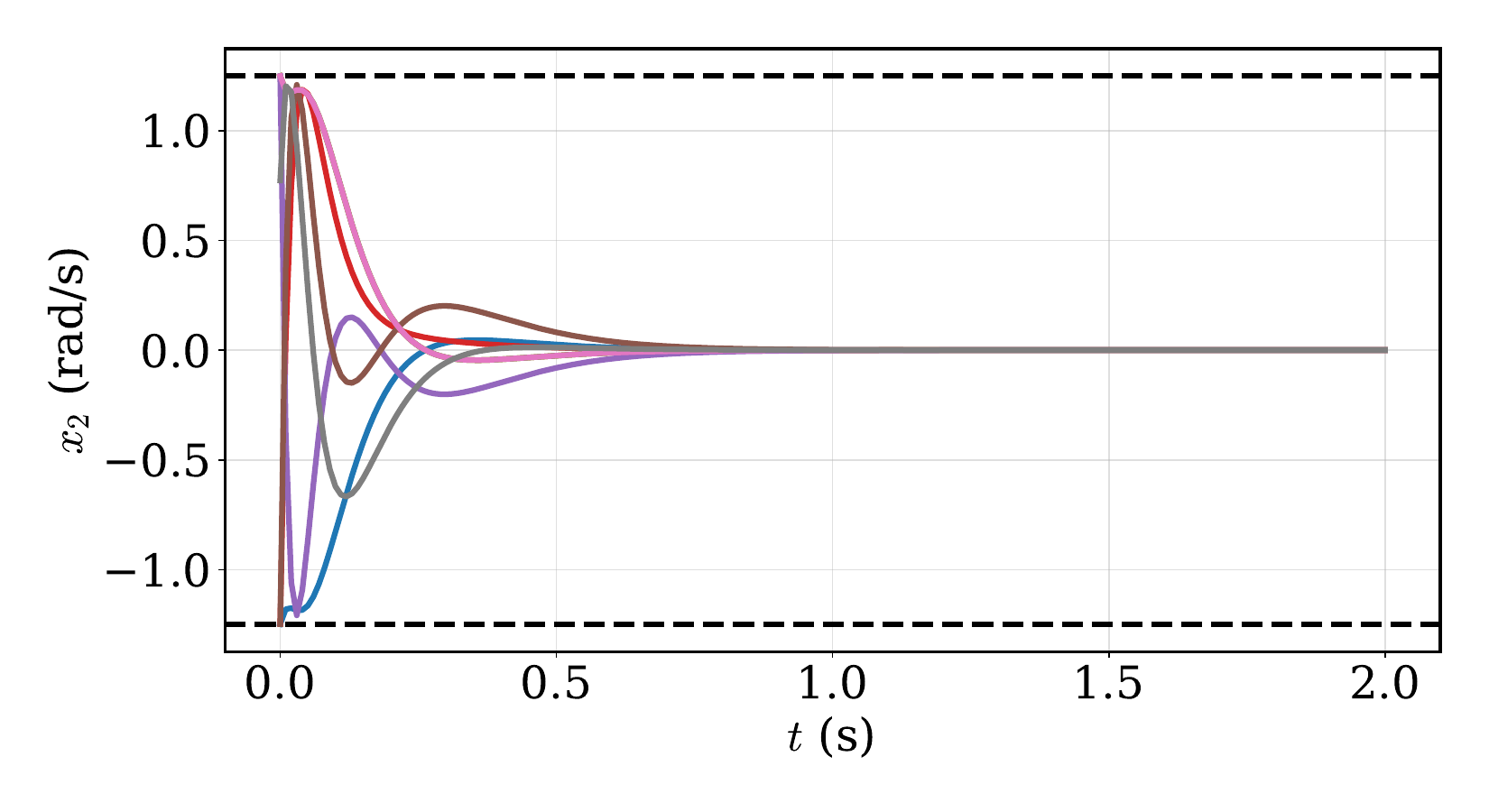}
  \caption{{Motor velocity $x_2$. Dashed: certified bound
  {$\pm1.2479$}~rad/s. This is the binding constraint for this template.}}
  \label{fig:sj_x2}
  \vspace{-0.6\baselineskip}
\end{figure}

\begin{figure}[!t]
  \centering
  \includegraphics[width=\linewidth]{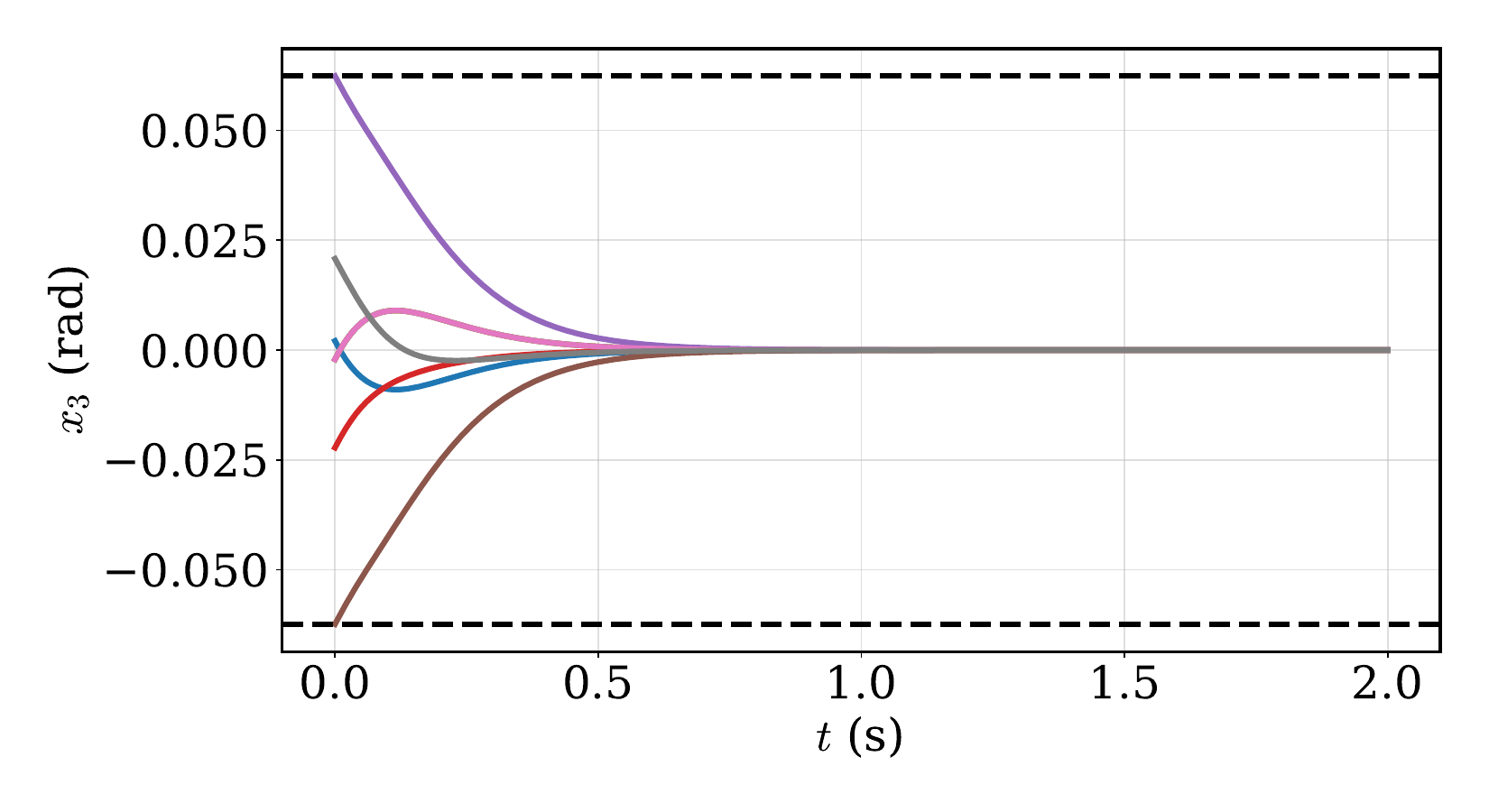}
  \caption{{Link angle $x_3$. Dashed: certified bound
  $\pm0.0624$~rad.}}
  \label{fig:sj_x3}
  \vspace{-0.6\baselineskip}
\end{figure}

\begin{figure}[!t]
  \centering
  \includegraphics[width=\linewidth]{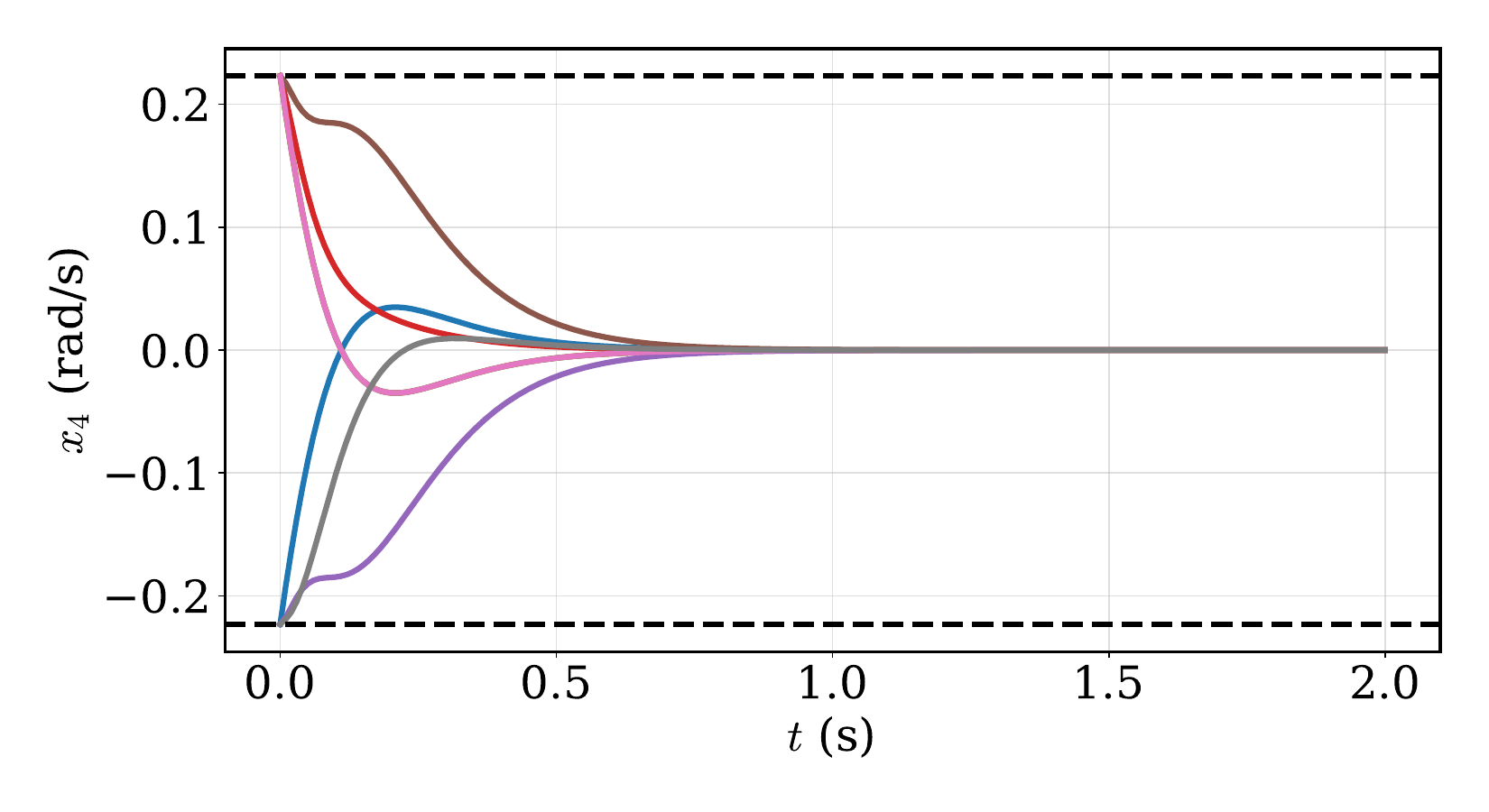}
  \caption{{Link velocity $x_4$. Dashed: certified bound
  {$\pm0.2228$}~rad/s.}}
  \label{fig:sj_x4}
  \vspace{-0.6\baselineskip}
\end{figure}



\section{Conclusion}
We proposed a data-driven safe control framework for nonlinear discrete-time
systems with uncertainty and disturbances, certifying invariance of a prescribed
polytopic safe set without canceling nonlinearities or resorting to worst-case
bounds. Convexity is preserved through a convexification--verification pipeline
based on {a control-oriented quadratic curvature compensation of the closed-loop facet maps}, enabling tractable synthesis and geometry-aware verification, while
{cell-wise curvature compensation and node-dependent interpolated controllers} reduce conservatism and
enlarge the class of certifiable invariant sets. {Future work will address maximal
invariant sets within prescribed templates, measurement noise, more efficient set
enlargement, and nonlinear basis function learning---the latter enabling
representations that satisfy the convexifiability conditions of Proposition~1
while bounding the remaining terms to preserve safety. These directions are
orthogonal to the proposed pipeline and will further improve applicability,
including terminal set design for MPC.}

\section{Appendix}
\noindent \textit{Proof of Theorem 1:} 
Under \eqref{cor:cont3}, based on Lemma 1, the closed-loop dynamics \eqref{data-f} is  
\(
x(t+1)
= X_1G_{K,1}x(t) + \delta(t), \label{cor:cl}
\)
where 
\(
\delta(t)=X_1 G_{K,2}Q(x(t)) -W_0 \big(G_{K,1}x+G_{K,2}Q(x(t))\big)  + w(t),
\)
collects the nonlinear remainder, the disturbance terms
induced by $W_0$ and $w(t)$. To formalize $\lambda$-contractivity, fix a facet $i\in \mathcal I_s$ and
consider the worst-case one-step evolution along this facet as $\gamma_i
= \max_{\substack{ x\in\cal S(F,g)\\ W_0\in\mathcal W_W}}
F_{i,:}\,x(t+1)$, which becomes
\begin{align}
 \gamma_i
&= \max_{\substack{ x\in\cal S(F,g)\\ W_0\in\mathcal W_W}}
\Big( F_{i,:}X_1G_{K,1}x + F_{i,:}\delta(t) \Big), \label{cor:gamma_def}
\end{align}
where
$\mathcal W_W=\{W\in\mathbb R^{n\times T}:\|W\|\le Th_w\}$. By {Definition~\ref{def:sets}}, $\lambda$-contractivity of $\cal S(F,g)$ is guaranteed if
\begin{align}
\gamma_i \le \lambda g_i, \qquad i \in \mathcal I_s. \label{cor:lambda_req}
\end{align}
Define $f_{1,i}(x)=F_{i,:}X_1G_{K,1}x$ and
$f_{2,i}(x,W_0)=F_{i,:}\delta(t)$. Using the basic inequality
\(
\max (f_{1,i}+f_{2,i}) \;\le\; \max f_{1,i} + \max f_{2,i},
\)
a sufficient condition for \eqref{cor:lambda_req} is
\begin{align}
\underbrace{\max_{x\in\cal S(F,g)} F_{i,:}X_1G_{K,1}x}_{=~\bar\gamma_i}
+
\underbrace{\max_{\substack{x\in\cal S(F,g)\\ W_0\in\mathcal W_W}}
F_{i,:}\delta(t)}_{=~\bar\ell^{dw}_i}
\;\le\; \lambda g_i. \label{cor:suff_split}
\end{align}
The term $\bar\ell^{dw}_i$ depends on $G_{K,2}$ through $\delta(t)$.
Following the disturbance-cancellation/minimization philosophy, we upper-bound
$\bar\ell^{dw}_i$ by a slack variable $\eta_i$ and minimize this bound. In
particular,
\begin{align}
& \bar\ell^{dw}_i \le\;
 M_x \, L_Q
 \,\|F_{i,:}X_1G_{K,2}\|_2 + T h_w\, M_x  \,\|F_{i,:}\|_1\nonumber\\
& \times 
\big(\|G_{K,1}\|_2 + \|G_{K,2}\|_2\,{L_Q}\big)
+ \|F_{i,:}\|_1\,h_w
\le \eta_i , \label{eta5}
\end{align}
and the vector $\eta=[\eta_1,\ldots,\eta_s]^\top$ is minimized via the objective
\eqref{eta}.
Hence, \eqref{cor:suff_split} is satisfied if
\begin{align}
\bar\gamma_i +\eta_i \le \lambda g_i. \label{cor:suff_final}
\end{align}

Next, note that $\bar\gamma_i$ is the optimal value of the linear program
$\bar\gamma_i = \max_x F_{i,:}X_1G_{K,1}x \;\text{s.t.}\; Fx \le g$, with the dual \cite{SetB}
\begin{subequations}
\begin{align}
\hat\gamma_i
&=\min_{\alpha_i}\; \alpha_i^\top g
\label{cor:dual1}\\
&\text{s.t.}\quad \alpha_i^\top F = F_{i,:}X_1G_{K,1}, \,\, \alpha_i^\top \ge 0, \label{cor:dual3}
\end{align}
\end{subequations}
where $\alpha_i\in\mathbb R^s$ is the dual variable. {By weak duality, $\bar\gamma_i\le\alpha_i^\top g$ for every $\alpha_i$ feasible for \eqref{cor:dual3}.}  Now, stack the row vectors $\alpha_i^\top$ into 
$P_s
=
\begin{bmatrix}
\alpha_1 & \alpha_2 & \cdots & \alpha_s
\end{bmatrix}^{\!\top}
\in \mathbb R^{s\times s}.$
Then, \eqref{cor:dual3} implies $P_s\ge 0$, and  for all
$i \in \mathcal I_s$ is equivalently written as
\begin{align}
P_sF = FX_1G_{K,1}. \label{cor:PsF}
\end{align}
Finally, using \eqref{cor:dual1}, the condition \eqref{cor:suff_final} for all facets $i$ {is implied by}
\begin{align}
\alpha_i^\top g + \eta_i \le \lambda g_i,\qquad i \in \mathcal I_s,
\end{align}
which stacks to
\begin{align}
P_s g+\eta \le \lambda g. \label{cor:Psg}
\end{align}
Together with the data-consistency constraint $V_0 \, [G_{K,1} \quad G_{K,2}]=I$, \eqref{eta5}, \eqref{cor:PsF} and \eqref{cor:Psg}  yields
\eqref{cor:cont1}--\eqref{cor:cont4-1}. Therefore, if
\eqref{cor:cont1}--\eqref{cor:cont4-1} hold, then \eqref{cor:lambda_req} holds,
and the safe set $\cal S(F,g)$ is $\lambda$-contractive (hence robustly
invariant). \hfill $\blacksquare$ \vspace{3pt}

\noindent\textit{Proof of Theorem~2:}
{
By Lemma~1 and \eqref{cor:cont4}, $\lambda$-contractivity is equivalent to
$H_i(x)=c_i^\top x+h_i(x)\le\lambda g_i$ for all $x\in\cal S(F,g)$ and all
$i\in\mathcal I_s$, with $c_i$ and $h_i$ as defined in \eqref{ci}. Condition
\eqref{eq:eps_hess_matrix} makes
$\tilde h_i(x)=h_i(x)+\tfrac{\varepsilon_i}{2}\|x\|_2^2$ convex on
$\cal S(F,g)${, and since $\tfrac{\varepsilon_i}{2}\|x\|_2^2\ge 0$,}
\begin{align}
c_i^\top x+h_i(x)\;\le\; c_i^\top x+\tilde h_i(x)=:\Psi_i(x),
\, \forall x\in\cal S(F,g).
\label{eq:t2_major}
\end{align}
The function $\Psi_i$ is convex on $\cal S(F,g)$, and $\cal S(F,g)$ is a polytope, so
$\cal S(F,g)=\operatorname{conv}\{x_\ell\}_{\ell\in\mathcal I_\ell}$ and the
maximum of $\Psi_i$ is attained at a vertex,
\begin{align}
\max_{x\in\cal S(F,g)}\Psi_i(x)=\max_{\ell\in\mathcal I_\ell}\Psi_i(x_\ell).
\label{eq:t2_vmax}
\end{align}
For each $\ell\in\mathcal I_\ell$,
\(
\Psi_i(x_\ell)=F_{i,:}X_1G_{K,1}x_\ell+F_{i,:}X_1G_{K,2}q_\ell
+\tfrac12\varepsilon_i\|x_\ell\|_2^2\le\lambda g_i
\)
is precisely \eqref{eq:margin_merged_opt}. Combining this with
\eqref{eq:t2_vmax} and \eqref{eq:t2_major} yields $H_i(x)\le\lambda g_i$ for
all $x\in\cal S(F,g)$ and all $i\in\mathcal I_s$, which establishes
$\lambda$-contractivity.}
This completes the proof. \hfill$\blacksquare$

\noindent \textit{Proof of Corollary 1:}
{Since $0\in\mathrm{int}\,\cal S(F,g)$, we have $F_{i,:}0=0<g_i$ for all $i\in\mathcal I_s$, so $g>0$ and $V$ in \eqref{eq:V_def_cor} is well defined. In fact, $V$ is the gauge function of the polytope $\cal S(F,g)$: it satisfies $V(0)=0$ and $V(\alpha x)=\alpha V(x)$ for all $\alpha\ge0$, and, since $V(x)\le1$ if and only if $Fx\le g$, $\cal S(F,g)=\{x:V(x)\le1\}$. Moreover, $V(x)>0$ for every $x\neq0$. Indeed, if $V(x)\le0$ for some $x\neq0$, then $F(\alpha x)\le0\le g$ for every $\alpha\ge0$, so the whole ray $\{\alpha x\}_{\alpha\ge0}$ would lie in $\cal S(F,g)$, contradicting its boundedness. Hence, $V$ is positive definite on $\cal S(F,g)$.

To show that $V$ decreases along closed-loop trajectories, we use the convexified facet map $\Psi_i(x)=c_i^\top x+h_i(x)+\tfrac{\varepsilon_i}{2}\|x\|_2^2$ introduced in \eqref{eq:t2_major}. As shown in the proof of Theorem~2, $\Psi_i$ is convex on $\cal S(F,g)$, and by \eqref{eq:t2_vmax} and \eqref{eq:margin_merged_opt},
\begin{align}
\Psi_i(z)\le\lambda g_i,\qquad \forall z\in\cal S(F,g),\ \forall i\in\mathcal I_s.
\label{eq:cor_psi_bound}
\end{align}
In addition, since $Q(0)=0$, we have $h_i(0)=F_{i,:}X_1G_{K,2}Q(0)=0$, and therefore $\Psi_i(0)=0$. The same fact gives $H_i(0)=0$ for all $i$, so $x(t)=0$ implies $x(t+1)=0$ and the claim holds trivially at the origin.

Now fix $x(t)\in\cal S(F,g)\setminus\{0\}$, and let $\varepsilon=V(x(t))\in(0,1]$ and $\tilde x=x(t)/\varepsilon$. By the homogeneity of $V$, $V(\tilde x)=1$, so that $\tilde x\in\partial\cal S(F,g)\subseteq\cal S(F,g)$. Since $x(t)=\varepsilon\tilde x+(1-\varepsilon)\,0$ is a convex combination of the two points $\tilde x$ and $0$ of $\cal S(F,g)$, the convexity of $\Psi_i$, together with $\Psi_i(0)=0$ and \eqref{eq:cor_psi_bound}, gives
\begin{align}
\Psi_i(x(t))\le\varepsilon\Psi_i(\tilde x)+(1-\varepsilon)\Psi_i(0)\le\varepsilon\lambda g_i=\lambda V(x(t))\,g_i .
\label{eq:cor_jensen}
\end{align}
Combining \eqref{eq:t2_major} with \eqref{eq:cor_jensen}, for every $i\in\mathcal I_s$ one obtains
\(
F_{i,:}x(t+1)=H_i(x(t))\le\Psi_i(x(t))\le\lambda V(x(t))\,g_i .
\)
Dividing by $g_i>0$ and taking the maximum over $i$ yields, $\forall x(t)\in\cal S(F,g)$,
\begin{align}
V(x(t+1))=\max_{i\in\mathcal I_s}\frac{F_{i,:}x(t+1)}{g_i}\le\lambda V(x(t)).
\label{eq:cor_decrease}
\end{align}
Since $\cal S(F,g)$ is invariant by Theorem~2, \eqref{eq:cor_decrease} holds along the entire trajectory, so that $V(x(t))\le\lambda^tV(x(0))$. Thus, this shows that $V$ is a Lyapunov function for the closed-loop system. For $\lambda<1$, the origin is exponentially stable with $\cal S(F,g)$ contained in its region of attraction, and for $\lambda=1$, $V$ is nonincreasing.
}
\hfill $\blacksquare$\vspace{6pt}

\noindent \textit{Proof of Proposition 1:}
Fix any facet index $i\in\mathcal I_s$. From \eqref{eq:hyb_split} and \eqref{Hf},
for all $x\in\cal S(F,g)$, setting $h_i=z_i+r_i$, which is imposed by \eqref{eq:hyb_split}, one has
\begin{alignat}{2}
\label{eq:hyb_H_split_pf}
H_i(x)
=\underbrace{c_i^\top x+z_i^\top Q(x)}_{=H_i^{\mathrm{cvx}}(x)}
+\underbrace{r_i^\top Q(x)}_{=H_i^{\mathrm{res}}(x)}.
\end{alignat}

By \eqref{eq:hyb_convex}, for each fixed $i$ the map
\(
H_i^{\mathrm{cvx}}(x)=c_i^\top x+z_i^\top Q(x)
\)
is convex on $\cal S(F,g)$. Since $\cal S(F,g)$ is a polytope, there exists a vertex
$x_{\ell^\star}$ such that
\begin{align} \label{Hpro1}
  \max_{x\in\cal S(F,g)} H_i^{\mathrm{cvx}}(x)
=\max_{\ell\in\mathcal I_\ell} H_i^{\mathrm{cvx}}(x_{\ell})
=
c_i^\top x_{\ell^\star}+z_i^\top q_{\ell^\star}.  
\end{align}

Moreover, for each $i\in\mathcal I_s$
there exists $\alpha_i\ge 0$ such that $c_i^\top=\alpha_i^\top F$. Hence, for any
vertex $x_\ell$ one has 
\begin{align}
  c_i^\top x_\ell
=\alpha_i^\top F x_\ell
=\alpha_i^\top g-\alpha_i^\top(g-Fx_\ell),
\, g-Fx_\ell\succeq 0.  
\end{align}
Let $x_{\ell^\star}\in\arg\max_{\ell\in\mathcal I_\ell}\big(c_i^\top x_\ell+z_i^\top q_\ell\big)$.
Substituting the above identity at $\ell^\star$ in \eqref{Hpro1} yields
\begin{align}\label{conpart}
& \max_{x\in\cal S(F,g)} H_i^{\mathrm{cvx}}(x)
=
\alpha_i^\top g
+
\Big(z_i^\top q_{\ell^\star}-\alpha_i^\top(g-Fx_{\ell^\star})\Big)
\nonumber\\ &
=
\alpha_i^\top g
+
\max_{\ell\in\mathcal I_\ell}
\Big(z_i^\top q_\ell-\alpha_i^\top(g-Fx_\ell)\Big).
\end{align}
In particular, the convex component is maximized {exactly} at a vertex and no
curvature slack in introduced. Moreover,  splitting conservativeness is avoided, as both linear and nonlinear terms are evaluated at the same maximizing vertex.


We now treat the remaining non-convex nonlinearities. For $x_0=0$, we have $q_0=Q(0)=0$. Therefore, by Assumption 6, one has
\(
r_i^\top Q(x)
\le \|L_Q\odot r_i\|_1\,R_0.
\) By \eqref{eq:hyb_abs_r}, ${k_i}\ge |L_Q\odot r_i|$ componentwise, so
$\|L_Q\odot r_i\|_1\le \mathbf 1^\top {k_i}$, and thus
\begin{alignat}{2}
\label{eq:hyb_res_lip_s}
H_i^{\mathrm{res}}(x)\le R_0 \,\mathbf 1^\top {k_i},
\qquad \forall x\in\cal S(F,g),\ \forall i\in\mathcal I_s.
\end{alignat}

Combining \eqref{eq:hyb_H_split_pf}, \eqref{conpart}, and
\eqref{eq:hyb_res_lip_s}, for any fixed vertex $x_\ell$ and all $x\in\cal S(F,g)$,
\(
H_i(x)\le \alpha_i^\top g
+ \big(z_i^\top q_\ell-\alpha_i^\top(g-Fx_\ell)\big)
+ R_0\,\mathbf 1^\top {k_i} .
\)
Then, stacking $\alpha_i^\top$ row-wise into $P\in\mathbb R^{s\times s}$, i.e.,
$P_{i,:}=\alpha_i^\top$, the constraints \eqref{eq:hyb_epi} and \eqref{eq:hyb_margin}
together enforce
\begin{align}
& H_i(x) \;\le\;(Pg)_i+t_i\;\le  \;\lambda g_i, \, \forall x\in\cal S(F,g), \, \forall i\in\mathcal I_s,    
\end{align}
This yields
$Fx\le g\Rightarrow F x(t{+}1)\le {Pg+t}\le \lambda g$. 

We now show that this hybrid certificate is no more conservative than Theorem~1
and can strictly outperform it. A Lipschitz-only (Theorem~1-type) verification
bounds the entire nonlinear term $h_i^\top Q(x)$ using
\eqref{eq:comp_lip_assm}: for any vertex $x_\ell$ and all $x\in\cal S(F,g)$,
$h_i^\top Q(x)
\le  R_0 \,\|L_Q\odot h_i\|_1.$
In the hybrid construction, \eqref{eq:hyb_abs_b} implies
$a_i\ge |L_Q\odot h_i|$ componentwise, hence
$\mathbf 1^\top a_i\ge \|L_Q\odot h_i\|_1$, while \eqref{eq:hyb_dominate} enforces
$\mathbf 1^\top k_i\le \mathbf 1^\top a_i$; therefore
$R_0\,\mathbf 1^\top k_i \le R_0\,\mathbf 1^\top a_i$, so
$R_0\,\mathbf 1^\top k_i$ is an admissible upper bound on
$R_0\,\|L_Q\odot h_i\|_1$. The portion assigned to the convex term
$z_i^\top Q(x)$ is instead handled {exactly} by vertex maximization. Hence,
relative to a pure Lipschitz treatment, moving any contribution of
$h_i^\top Q(x)$ from the residual $r_i^\top Q(x)$ into the convex part
$z_i^\top Q(x)$ can only reduce the overall upper bound, and whenever the
optimization convexifies a nonzero portion (i.e., $z_i\neq 0$), the certificate
strictly dominates the Lipschitz-based bound of Theorem~1. 
\hfill$\blacksquare$ \vspace{6pt}

{
\noindent \textit{Proof of Theorem~\ref{thm:th3_face_restricted_relaxed}:}
Let $(\{g_\ell\},G_{K,2},\{E_{c,i}\})$ be feasible for
\eqref{eq:th3_face_opt_relaxed}. In the absence of disturbances the data
satisfy $X_1=\bar A_1X_0+A_2Q(X_0)+BU_0=\begin{bmatrix}\bar A_1&A_2\end{bmatrix}V_0+BU_0$.
Multiplying on the right by $g_\ell$ and by $G_{K,2}$ and using
\eqref{eq:th3_data_relaxed} gives
\begin{align}
X_1g_\ell=\bar A_1x_\ell+Bu_\ell,
\qquad
X_1G_{K,2}=A_2+BK_2 .
\label{eq:th3_data_identities}
\end{align}
Fix $x\in\cal S(F,g)$, a cell $\mathcal C_c\ni x$, which exists by
\eqref{eq:th3_cells}, and weights $\theta(x)$ satisfying
\eqref{eq:th3_theta_relaxed}. Under \eqref{eq:th3_policy_relaxed}, using
$x=\sum_{\ell\in{\mathcal J_c}}\theta_\ell(x)x_\ell$ and
\eqref{eq:th3_data_identities},
\begin{align}
x(t+1)
&=\bar A_1x+A_2Q(x)+B\Big(\sum_{\ell\in{\mathcal J_c}}\theta_\ell(x)u_\ell+K_2Q(x)\Big)
\nonumber\\
&=\sum_{\ell\in{\mathcal J_c}}\theta_\ell(x)\big(\bar A_1x_\ell+Bu_\ell\big)
+(A_2+BK_2)Q(x)
\nonumber\\
&=\sum_{\ell\in{\mathcal J_c}}\theta_\ell(x)\,X_1g_\ell+X_1G_{K,2}Q(x).
\label{eq:th3_cl}
\end{align}
Hence, for every $i\in\mathcal I_s$, by \eqref{eq:th3_H_relaxed},
\begin{align}
F_{i,:}x(t+1)
=\sum_{\ell\in{\mathcal J_c}}\theta_\ell(x)\,F_{i,:}X_1g_\ell+h_i(x).
\label{eq:th3_facet_cl}
\end{align}
Define the convexified function
\begin{align}
\widetilde h_{c,i}(z)
=h_i(z)+\tfrac12\,(z-\bar x_c)^\top E_{c,i}\,(z-\bar x_c).
\label{eq:th3_Htilde_relaxed}
\end{align}
By \eqref{eq:th3_soc_face_relaxed},
$\partial^2\widetilde h_{c,i}(z)/\partial z^2
=\partial^2h_i(z)/\partial z^2+E_{c,i}\succeq0$ for all $z\in\mathcal C_c$,
so $\widetilde h_{c,i}$ is convex on the convex set $\mathcal C_c$. Since
$x_\ell\in\mathcal C_c$ for all $\ell\in{\mathcal J_c}$ and $x$ is the
convex combination \eqref{eq:th3_theta_relaxed} of these points, Jensen's
inequality gives
\begin{align}
\widetilde h_{c,i}(x)
\le\sum_{\ell\in{\mathcal J_c}}\theta_\ell(x)\,\widetilde h_{c,i}(x_\ell).
\label{eq:th3_jensen}
\end{align}
Since $E_{c,i}\succeq0$ by \eqref{eq:th3_adm_relaxed}, the quadratic term in
\eqref{eq:th3_Htilde_relaxed} is nonnegative, whence
$h_i(x)\le\widetilde h_{c,i}(x)$. Combining this with
\eqref{eq:th3_facet_cl}, \eqref{eq:th3_jensen} and
$\widetilde h_{c,i}(x_\ell)=F_{i,:}X_1G_{K,2}\,q_\ell
+\tfrac12(x_\ell-\bar x_c)^\top E_{c,i}(x_\ell-\bar x_c)$,
\begin{align}
F_{i,:}x(t+1)
&\le\sum_{\ell\in{\mathcal J_c}}\theta_\ell(x)
\Big(F_{i,:}X_1g_\ell+F_{i,:}X_1G_{K,2}\,q_\ell
\nonumber\\
&\qquad
+\tfrac12(x_\ell-\bar x_c)^\top E_{c,i}(x_\ell-\bar x_c)\Big)
\nonumber\\
&\le\sum_{\ell\in{\mathcal J_c}}\theta_\ell(x)\,\lambda g_i=\lambda g_i,
\label{eq:th3_final}
\end{align}
where the last inequality is \eqref{eq:th3_budget_relaxed}, written for the
cell $c$ and each of its nodes, multiplied by $\theta_\ell(x)\ge0$, and the
equality uses $\sum_{\ell\in{\mathcal J_c}}\theta_\ell(x)=1$. As
\eqref{eq:th3_final} holds for every $i\in\mathcal I_s$,
$x(t+1)\in\lambda\cal S(F,g)$. Since $x\in\cal S(F,g)$, the cell containing
it and the weights were arbitrary, $\cal S(F,g)$ is $\lambda$-contractive.
\hfill$\blacksquare$
}%

\vspace{3pt}

\noindent \textit{Proof of Theorem 4.}
Fix a facet index $i\in \mathcal I_s$. For any $x\in\cal S(F,g)$ and
$w\in\mathcal W$, the $i$th facet map of the disturbed closed-loop update is
\begin{alignat}{2}
& H_i^w(x,w)
= \nonumber \\ & F_{i,:}\Big((X_1-W_0)G_{K,1}x+(X_1-W_0)G_{K,2}Q(x)+w\Big) \nonumber\\
&= F_{i,:}X_1G_{K,1}x + F_{i,:}X_1G_{K,2}Q(x) \nonumber \\ &
   -F_{i,:}W_0\big(G_{K,1}x+G_{K,2}Q(x)\big) + F_{i,:}w \label{Hw1_relaxed}.
\end{alignat}
{
Every column of $W_0$ lies in $\mathcal W$, so every entry of the row vector
$F_{i,:}W_0\in\mathbb R^{1\times T}$ is bounded in magnitude by
$h_w\|F_{i,:}\|_1$, whence $\|F_{i,:}W_0\|_2\le\sqrt T\,h_w\|F_{i,:}\|_1$;
likewise $F_{i,:}w\le h_w\|F_{i,:}\|_1$. Therefore
\begin{alignat}{2}
& -F_{i,:}W_0\big(G_{K,1}x+G_{K,2}Q(x)\big)+F_{i,:}w
\le \nonumber \\ & h_w\|F_{i,:}\|_1\Big(\sqrt T\big(\|G_{K,1}x\|_2 + \|G_{K,2}Q(x)\|_2\big)+1\Big).
\label{eq:W0_support_relaxed}
\end{alignat}
The map $x\mapsto\|G_{K,1}x\|_2$ is convex, so its maximum over the polytope
$\cal S(F,g)$ is attained at a vertex and the first inequality in
\eqref{eq:robust_epi_norms_relaxed} gives $\|G_{K,1}x\|_2\le\xi^{(1)}$ for all
$x\in\cal S(F,g)$. The map $x\mapsto\|G_{K,2}Q(x)\|_2$ is not convex in
general, and is therefore bounded column-wise: by the triangle inequality and
the definition of $\bar q_k$,
\begin{alignat}{2}
& \|G_{K,2}Q(x)\|_2
\le\sum_{k=1}^{N}|Q_k(x)|\,\big\|(G_{K,2})_{:,k}\big\|_2
\nonumber \\ & \qquad
\le\sum_{k=1}^{N}\bar q_k\,\big\|(G_{K,2})_{:,k}\big\|_2
\le\xi^{(2)},
\label{eq:G2Q_bound_relaxed}
\end{alignat}
for all $x\in\cal S(F,g)$, where the last inequality is the second condition in
\eqref{eq:robust_epi_norms_relaxed}. Combining
\eqref{Hw1_relaxed}--\eqref{eq:G2Q_bound_relaxed}, for all $x\in\cal S(F,g)$
and all $w\in\mathcal W$,
\begin{alignat}{2}
& H_i^w(x,w)
\le \widehat H_i(x)
= c_i^\top x + h_i(x) + (d_T)_i,
\label{eq:Hhat_def_relaxed}
\end{alignat}
where $c_i^\top = F_{i,:}X_1G_{K,1}$, $h_i(x) = F_{i,:}X_1G_{K,2}Q(x)$ and
$(d_T)_i=h_w\|F_{i,:}\|_1\big(\sqrt T(\xi^{(1)}+\xi^{(2)})+1\big)$ does not
depend on $x$. Therefore, a sufficient condition for robustness
 is
\begin{alignat}{2}
\max_{x\in\cal S(F,g)} \widehat H_i(x) \le \lambda g_i,
\qquad i \in \mathcal I_s.
\label{eq:suff_goal_relaxed}
\end{alignat}
Condition \eqref{eq:robust_hess_relaxed} reads
$\frac{\partial^2 h_i(x)}{\partial x^2} \succeq -\varepsilon_i I$ for all
$x\in\cal S(F,g)$. Define the convexified facet map
\begin{alignat}{2}
\widetilde H_i(x)
&= c_i^\top x + h_i(x) + (d_T)_i + \tfrac12\varepsilon_i\|x\|_2^2.
\label{eq:Htilde_def_relaxed}
\end{alignat}
By the same convexity argument as Theorem~2, $\widetilde H_i$ is convex on
$\cal S(F,g)$, so its maximum is attained at a vertex, and
$\widehat H_i(x)\le\widetilde H_i(x)$ since $\varepsilon_i\ge0$. Hence
\begin{alignat}{2}
& \max_{x\in\cal S(F,g)} \widehat H_i(x)
\le \nonumber \\ &
\max_{\ell \in \mathcal I_{\ell}}
\Big(c_i^\top x_\ell + h_i(x_\ell) + \tfrac12\varepsilon_i\|x_\ell\|_2^2\Big)+(d_T)_i
\le \lambda g_i,
\label{eq:max_at_vertices_relaxed}
\end{alignat}
where the last inequality is \eqref{eq:robust_budget_relaxed}. The linear and
convexified nonlinear contributions are thus bounded jointly at the same
vertex, and only the data-noise and disturbance terms are bounded through the
margin $(d_T)_i$, which is optimized jointly with the gains through
$\xi^{(1)}$ and $\xi^{(2)}$. Since $H_i^w(x,w)\le \widehat H_i(x)$ by
\eqref{eq:Hhat_def_relaxed}, we conclude $H_i^w(x,w)\le \lambda g_i$ for all
$x\in\cal S(F,g)$ and all $w\in\mathcal W$. Stacking all facets yields robust
$\lambda$-contractivity.
\hfill$\blacksquare$ \vspace{6pt}
}%


\begin{IEEEbiographynophoto}{Amir Modares}
 (Student Member, IEEE) received a bachelor’s degree in electrical engineering from
Azad University, Qaen, Iran, in 2019, and a master’s degree in electrical engineering from the Sharif University of Technology, Tehran, Iran, in 2022.
His research interests are machine learning, convex
optimization, safe control, and nonlinear control. He is now a PhD student at the University of Cyprus. 
\end{IEEEbiographynophoto} \vspace{-20pt}

\begin{IEEEbiographynophoto}{Bahare Kiumarsi}
received her B.S. degree in Electrical Engineering from Shahrood University of Technology, Shahrood, Iran, in 2009, her M.S. degree from Ferdowsi University of Mashhad, Iran, in 2013, and her Ph.D. degree from the University of Texas at Arlington, Arlington, TX, USA, in 2017. She is currently an Assistant Professor in the Department of Electrical and Computer Engineering at Michigan State University. Prior to joining Michigan State, she was a Postdoctoral Research Associate at the University of Illinois at Urbana-Champaign. Her research interests include learning-based control, the security of cyber-physical systems, and distributed control of multi-agent systems. She serves as an Associate Editor for Neurocomputing.
\end{IEEEbiographynophoto}\vspace{-25pt}

\begin{IEEEbiographynophoto}{Hamidreza Modares}
received a B.S. degree from the University of Tehran, Tehran, Iran, in 2004, an M.S. degree from the Shahrood University of Technology, Shahrood, Iran, in 2006, and a Ph.D. degree from the University of Texas at Arlington, Arlington, TX, USA, in 2015, all in Electrical Engineering. He is currently an Associate Professor in the Department of Mechanical Engineering at Michigan State University. Before joining Michigan State University, he was an Assistant Professor in the Department of Electrical Engineering at Missouri University of Science and Technology. His current research interests include reinforcement learning, safe control, machine learning in control, distributed control of multi-agent systems, and robotics. He is an Associate Editor of IEEE Transactions on Systems, Man, and Cybernetics: systems. 
\end{IEEEbiographynophoto}

\end{document}